\documentclass[conference]{IEEEtran}
\IEEEoverridecommandlockouts

\usepackage[font={small,it},skip=3pt,belowskip=5pt]{caption}
\usepackage{mathptmx}  
\usepackage{amsmath}
\usepackage{subcaption}
\usepackage{amsthm}
\usepackage{amsfonts}
\usepackage{amssymb}
\usepackage{graphicx}
\usepackage{url}
\usepackage{msc}
\usepackage{color}  
\usepackage[ruled,vlined,longend,linesnumbered]{algorithm2e}
\usepackage{algorithmicx}
\usepackage{paralist}
\usepackage{tikz}
\usepackage{booktabs}
\usepackage{hyperref}

\newcommand{\subparagraph}{}
\usepackage[compact]{titlesec}

\usetikzlibrary{decorations.pathreplacing}

\newcommand{\myparagraph}[1]{\noindent{\bf #1.}}
\renewcommand{\tablename}{Tab.}

\renewcommand{\figurename}{Fig.}
\newtheorem{theorem}{Theorem}[section]
\newtheorem{lemma}[theorem]{Lemma}
\newtheorem{definition}{Definition}
\newtheorem{corollary}{Corollary}

\newcommand{\name}{LaKSA\xspace}

\newcommand{\comb}[2]{\left( \hspace{-0.15cm} \begin{array}{c} #1 \\ #2 \end{array} \hspace{-0.15cm} \right)}
\newcommand{\advfrac}{\alpha}
\newcommand{\commitb}{B}

\newcommand{\cround}{m}

\renewcommand{\P}{\mathbb{P}}
\newcommand{\E}{\mathbb{E}}

\newcommand{\prole}{p_{\textsc{role}}}
\newcommand{\nrole}{N_{\textsc{role}}}
\algdef{SE}[DOWHILE]{Do}{doWhile}{\algorithmicdo}[1]{\algorithmicwhile\ #1}%

\makeatletter
\newcommand{\removelatexerror}{\let\@latex@error\@gobble}
\makeatother

\begin{document}
\setlength{\abovedisplayskip}{3pt}
\setlength{\belowdisplayskip}{3pt}
\setlength{\abovedisplayshortskip}{3pt}
\setlength{\belowdisplayshortskip}{3pt}

\hyphenation{a-na-lo-gous-ly}

\title{\name: A Probabilistic Proof-of-Stake Protocol
\thanks{We thank the anonymous reviewers for their valuable comments and
suggestions. This work was supported in part by the Ministry of Education,
Singapore, under its MOE AcRF Tier 2 grant (MOE2018-T2-1-111), by the
National Research Foundation (NRF), Prime Minister's Office, Singapore, under
its National Cybersecurity R\&D Programme (Award No. NRF2016NCR-NCR002-028) and
administered by the National Cybersecurity R\&D Directorate, and by
A*STAR under its RIE2020 Advanced Manufacturing and Engineering (AME)
Programmtic Programme  (Award A19E3b0099).}
}

 \newcommand{\affa}{$^*$}
 \newcommand{\affb}{$^\dag$}
 \newcommand{\auth}{$^\ddag$}
 \author{
 \IEEEauthorblockN{Dani\"el Reijsbergen\affa\auth, Pawel Szalachowski\affa\auth$^\#$, Junming Ke\affb, Zengpeng Li\affa, and Jianying Zhou\affa}
 \IEEEauthorblockA{\\[-0.2cm]\affa Singapore University of Technology and Design, Singapore}
 \IEEEauthorblockA{\affb University of Tartu, Tartu, Estonia}
 \IEEEauthorblockA{\auth \textit{The authors contributed equally to this work.}\quad $^\#$\textit{Now at Google.}}
 }

\maketitle

\begin{abstract}
    We present Large-scale Known-committee Stake-based Agreement (\name), a chain-based Proof-of-Stake protocol
    that is dedicated, but not limited, to cryptocurrencies.  \name minimizes interactions
    between nodes through lightweight committee voting, resulting in a simpler, more robust, and more scalable proposal
    than competing systems.  It also mitigates other drawbacks of previous
    systems,
    such as high reward variance and long confirmation times.
    \name can support large numbers of nodes by design, and provides
    probabilistic safety guarantees in which a client makes commit decisions by calculating the probability that a transaction is reverted based on its blockchain view.
    We present a thorough analysis of \name and
    report on its implementation and evaluation. Furthermore, our new technique of proving safety can be applied more broadly to other Proof-of-Stake protocols.
\end{abstract}

\section{Introduction}
\label{sec:intro}
One of the main innovations of Bitcoin~\cite{nakamoto2008bitcoin} is Nakamoto Consensus (NC), a
protocol for maintaining a distributed data structure called a \emph{blockchain}. In NC, each participating node tries to become the round leader by
solving a Proof-of-Work (PoW) puzzle.  If a solution is
found, the node announces a \emph{block} that contains a hash link to the previous block,
the solution of the PoW puzzle, transactions, and metadata.  Potential network forks are
resolved by following the \textit{longest-chain} rule.  In this
non-interactive leader election process,
a block acts as a medium of transactions, as a leader election
proof that is easily verifiable by other nodes, and as confirmation of all previous
blocks. Hence, confidence in a given chain is built gradually through the addition of new blocks.  The NC protocol can
scale to a number of nodes that is vastly too large to be handled by traditional
Byzantine Fault Tolerant (BFT) consensus
protocols~\cite{castro1999practical}, which in  every round require significant communication
complexity among all nodes.  However, NC has critical limitations: an enormous energy footprint, low throughput, and a slow and
insecure block commitment process that relies on the absence of large-scale adversarial behavior \cite{eyal2014majority}. Furthermore, its reward frequency and structure encourage 
centralization, which in turn amplifies its security
vulnerabilities~\cite{bonneau2015sok,croman2016scaling,gervais2016security}. 

A promising direction in the mitigation of (some of) these drawbacks involves
Proof-of-Stake (PoS) protocols, in which nodes do not vote with their computing power, but with their \emph{stake}, i.e., blockchain tokens. In PoS, new nodes need to obtain stake from existing nodes -- although this is more restrictive than Bitcoin's fully permissionless setting, PoS systems
promise multiple benefits. The main benefit is reduced energy consumption,
as voting is typically based on a cryptographic lottery  that relies only on short messages instead of resource-intensive PoW puzzles.  
Interestingly, as participating nodes and their stake are known upfront, these
systems also promise to commit blocks faster and with better safety properties, i.e.,
similar properties as those in standard BFT protocols.

However, there are many obstacles in the design of new PoS schemes.  Several
recent systems try to apply the longest-chain
rule in the PoS setting.  Unfortunately, those proposals mimic Nakamoto's design
decisions that introduce undesired effects, e.g., a high reward variance
which encourages centralization~\cite{fanti2018compounding}.  Moreover, although they are conceptually similar to NC, whose properties have been investigated in the literature, it has
proven challenging to replace PoW-based voting by cryptographic primitives while avoiding fundamental limitations~\cite{brown2018formal}.
For instance, in NC it is expensive to work on PoW
puzzles on two branches simultaneously because the node will eventually lose its
investment on the losing branch. However, it is trivial in PoS systems to vote on multiple
conflicting branches -- this is informally known as the \emph{nothing-at-stake} problem.  Similarly, it is easier to launch \textit{long-range attacks} in
the PoS setting, in which an adversary obtains keys of nodes who controlled a
majority of stake in the distant past and creates a long alternative chain.

In this work, we present the \textit{\underline{La}rge-scale \underline{K}nown-committee \underline{S}take-based \underline{A}greement} (\name) protocol. \name is a PoS protocol that addresses
the above limitations while remaining conceptually simple and free from BFT-like
message complexity.  In our protocol, committee members are pseudorandomly and periodically sampled
 to vote for their preferred views of the main chain.  Evidence for these votes is
 included in the blocks, and users make their own decisions about whether to commit a block and act on the block's transactions -- e.g., dispatch a package after receiving a cryptocurrency payment. In \name, clients
calculate the \emph{probability} that the block can be reverted given how many votes support the block and its descendants, and commit the block if this probability is low enough.
We show that a design with fixed-size pseudorandom committees
brings multiple benefits, and although it also introduces a vulnerability to adaptive attackers, we discuss how to mitigate this threat using network-level anonymization techniques from the existing literature \cite{bojja2017dandelion,fanti2018dandelion++}.

By reducing the degree of interaction between committee members, \name can scale to large
numbers of active participants. This is combined with a high frequency
of blocks and rewards for
participating nodes, which mitigates the tendency towards centralization in NC.
The chain selection algorithm in \name aims to gather as much information as possible about the nodes' views, which leads to higher security and faster commit decisions.
The CAP theorem~\cite{brewer2000towards} restricts distributed systems to select
either availability or consistency in the face of network partitions. \name
favors availability over consistency; however, in \name partitioned clients
would notice that their view of the blockchain does not provide strong safety
guarantees for their commit decisions -- this is demonstrated through a novel use of hypothesis
testing.
\name is also augmented with a fair and coalition-safe rewarding scheme,
omitted by many competing
systems~\cite{badertscher2018ouroboros,david2018ouroboros,gilad2017algorand}.

We emphasize the impact of the major design choices in \name through a detailed comparison to  Algorand \cite{gilad2017algorand}, which is a closely-related PoS protocol. Whereas Algorand's BFT-like block commitment scheme focuses on individual rounds, \name aggregates information from multiple rounds and hence achieves higher security and flexibility in the sense that each user is able to set her own security thresholds.
We thoroughly analyze our system,  present  the efficiency
and deployability of \name through an implementation, and discuss alternative design choices and
extensions. 

\myparagraph{Contributions}
\name introduces a novel chain selection mechanism and probabilistic commit
rules analyzed by a novel application of statistical hypothesis testing. To our best
knowledge, \name is the first concrete protocol with fixed-size pseudorandom
committees that demonstrates its advantages through a comparison to related work.
Finally, the cryptographic sampling procedure is also a new construction.

\section{Preliminaries \& Requirements}
\label{sec:pre}
\myparagraph{Network \& Threat Model}
Our system consists of multiple peer nodes who seek to achieve
consensus over the state of a distributed ledger.  The nodes are
mutually untrusting with exactly the same rights and functions -- i.e., no node or
any subset is trusted.  Nodes are identified by their unique public keys, and
to avoid Sybil attacks~\cite{douceur2002sybil} we assume that nodes require stake in terms of the underlying cryptocurrency to participate. We assume that the nodes' stakes are
integers that sum to the total stake $n$. The
initial set of nodes and their stake allocation is pre-defined in a genesis
block that is known to all nodes (i.e., the system is bootstrapped from the
genesis block). 
The stake distribution is represented by a
mapping between public keys and their corresponding stake units.  The system is
used by clients who wish to make cryptocurrency transactions and who do not necessarily participate in
the consensus protocol (although some may).

We assume that participating nodes have loosely synchronized clocks, as the
protocol makes progress in timed rounds.  Nodes are connected via a peer-to-peer
network in which messages are delivered within $\Delta$ seconds during periods of  network
synchrony.  We assume a partially synchronous communication
model~\cite{dwork1988consensus}, where the network can be asynchronous with
random network delays, but synchrony is eventually resumed after an unknown time period
called global stabilization time (GST). We restrict this further by assuming that during
asynchronous periods, an adversary cannot adaptively move nodes between splits without being detected, although the adversary can be present in all existing splits. 

The adversary in our model is able to control nodes possessing $f=\alpha n$ ``malicious''
stake \emph{units}, for any \mbox{$\alpha \in [0,\frac{1}{3})$} and where $n$ denotes the total number of stake units; hence, honest
nodes possess
$n-f=(1-\alpha)n$ stake units. Here, both $f$ and $n$ are integers. Adversarial nodes can be byzantine, i.e., they can fail
or equivocate messages at the adversary's will.  We assume that $\alpha<\frac{1}{3}$ is a
bound in both our adversary and network models~\cite{dwork1988consensus}.  We assume
that the adversary's goal is to a) stop the progress of the protocol, or b) cause a
non-compromised client to commit a transaction that is reverted with a
probability that is higher than a threshold $p^*$ specified by the client.  We also discuss
adversarial strategies that aim to undermine other protocol properties, such as
efficiency, throughput, and fairness.

In our model, honest nodes faithfully execute the protocol, but can be \textit{offline}, meaning that their blocks and messages cannot be received by the rest of network. A node can be offline due to a network fault, or because its client is inactive due to a software or hardware fault. Such faults are usually {temporary} but can be permanent -- e.g., if the node has lost the private key associated with its stake. For \name to satisfy the property of liveness as discussed below, we require honest nodes to be online after the GST. Any node that is offline after the GST is grouped with the adversarial nodes, in line with BFT protocols. The temporary unavailability of honest nodes slows down block commitment, because support for the blocks is accumulated more slowly if fewer votes are received by the block proposers. However, as we discuss in \autoref{sec:analysis}, offline nodes do not make a safety fault -- i.e., a committed block being reverted -- more likely. 
As such, we can safely assume that all nodes are online for the safety proof.

\myparagraph{Desired Properties}
First, we desire the two following fundamental properties.

\begin{compactitem}[]
    \item \textit{Liveness}: a valid transaction sent to the network is
        eventually added to the global ledger and committed by all clients.
    \item \textit{Probabilistic Safety}: if for a client-specified probability
        \mbox{$p^*\in[0,1]$}, the client commits a transaction, then the probability that this transaction is ever
        reverted is at most $p^*$.
\end{compactitem}
One important aspect of \name is that, unlike in traditional BFT protocols, our
definition of safety is probabilistic. The relaxed safety property allows us to
scale the system to thousands of active participants and to propose a simple,
robust, and high-throughput system, while still achieving strong client- or even
transaction-specific safety.

As we present our work in an open cryptocurrency setting, we also aim to achieve the following additional properties.
\begin{compactitem}[]
    \item \textit{Scalability}: the system should scale to large numbers of
        nodes.
    \item \textit{High Throughput}: the system should provide high throughput
        for transactions, and in particular, the consensus mechanism should not
        be a bottleneck for this throughput.
    \item \textit{Efficiency}: overheads introduced by the system should be
        reasonable, allowing system deployment on the Internet as it is today,
        without requiring powerful equipment or resources -- e.g., CPU, memory, and
        bandwidth.
    \item \textit{Fairness}: an honest node with a fraction $\beta$ of the total
        stake, should have a presence of approximately $\beta$ in the
        main chain of the blockchain. This is especially important when the
        presence in the blockchain is translated into system rewards.
    \item \textit{Coalition Safety}: any coalition of nodes with the total stake
        $\alpha=\sum \alpha_i$, where $\alpha_i$ is a coalition node's stake,
        cannot obtain more than a multiplicative factor $(1 + \epsilon)\alpha$
        of the total system rewards for some small $\epsilon$.
\end{compactitem}

\myparagraph{Cryptographic Notation}
We make standard cryptographic assumptions and we use the following
cryptographic constructions.
\begin{compactitem}[]
    \item $H(m)$ is a collision-resistant cryptographic hash function, producing
        a hash value for a message $m$;
    \item $PRF_k(m)$ is a keyed pseudorandom function (PRF), outputting a
        pseudorandom string for key $k$ and message $m$;
    \item $Sign_{sk}(m)$ is a signature scheme that for a secret key $sk$ and a
        message $m$ produces the corresponding signature $\sigma$;
    \item $\textit{VrfySign}_{pk}(m, \sigma)$ is a signature verification
        procedure that returns \textit{True} if $\sigma$ is the correct signature of
        $m$ under the secret key corresponding to the public key $pk$, and \textit{False}
        otherwise.
\end{compactitem}

\section{Protocol}
\label{sec:details}
\subsection{Blockchain Structure}

\name operates through a blockchain, such that each block contains a set of transactions,
a link to the previous block, and various metadata. Recall that for each node we have a secret/public key pair $sk, pk$ as well as the amount of associated stake.  
The structure of a block is
\begin{equation}
    \label{eq:block}
    B = (i, r_i, H(B_{-1}), F, V, \mathit{Txs}, pk, \sigma),
\end{equation}
where
\begin{compactitem}[]
\item $i$ is the round number (consecutive, starting from 0);
    \item $r_i$ is a random value generated by the leader;
    \item $H(B_{-1})$ is the hash of the previous valid block through which blocks
        encode the \textit{parent-child} relationship;
    \item $V$ is a set of \textit{votes}
        that support the previous block (see \autoref{eq:vote});
    \item $F$ is a set of known, to the leader, forked blocks
        that have not been reported in any previous known block;
    \item $\mathit{Txs}$ is a set of transactions included in the block;
    \item $pk$ is the leader's public key;
    \item $\sigma$ is a signature, created by the leader over all
        previous fields except $pk$.
\end{compactitem}

Every block $B$ supports its predecessor $B_{-1}$ by including votes of nodes who were elected in $B$'s round and who vouched for $B_{-1}$ as the last block on their preferred chain.  A vote has the following structure:
    \begin{equation}\label{eq:vote}
        v=(i, H(B_{-1}), s, pk, \sigma),
    \end{equation}
where
\begin{compactitem}[]
    \item $i$ is the round number;
    \item $H(B_{-1})$ is the hash of the previous valid block;
    \item $s$ is the stake that the vote creator was elected with as a voter for
        the round $i$;
    \item $pk$ is the voter's public key;
    \item $\sigma$ is a signature, created by the voter over all
        previous fields except $pk$.
\end{compactitem}
Essentially, a vote encodes a stake unit with which the voter supports her blockchain view in a given round.

As our blockchain contains blocks that follow a parent-child relation and as it
may contain forks, the final rendered data structure is essentially a tree (see an example in \autoref{fig:chain}).  However,
within this tree only one branch is considered as the \textit{main chain} whose
transactions are strictly ordered.  Transactions are initiated by blockchain
nodes who wish to transfer their crypto tokens or execute a higher-level logic
(e.g., smart contracts).  Transactions typically also include fees
paid to round leaders for appending them to the blockchain.  
We do not discuss any concrete
transaction model or a validation logic for transactions, although
models used in other systems can easily be implemented with \name.

\subsection{Voting Round}
\label{sec:details:voting}

The protocol bootstraps from the genesis block and makes progress in two-step
rounds. The two steps each last $\Delta$ seconds, where $\Delta$ as defined in \autoref{sec:pre} -- for brevity, we treat $\Delta$ as a bound on all delays including message generation and
processing times. The voting procedure in each round is presented in
\autoref{alg:voting}.  At the beginning of round $i$, each node
obtains the round's pseudorandom beacon $r$ and determines the voters and leaders.
In \autoref{sec:details:beacon} we show concrete instantiations of beacon
generation and discuss alternative ways of realizing it.

        In the first step of any round $i$, each node checks whether it can vote in
        round $i$ by calling $VoterStake()$, which returns the number of stake units that it
        can use to vote in round $i$. If a positive number is returned, then the
        node is called a \textit{voter} in round $i$ and it can
        vote for the last block of what it believes to be the main chain to support this chain.  To do so, it creates a vote $v$ (see \autoref{eq:vote})
        and broadcasts the vote immediately to the network.  Other nodes validate each received vote by checking
        whether
        \begin{inparaenum}[a)]
            \item it is authentic, formatted correctly, and not from the future,
                 i.e., not with a round number that exceeds $i$,
            \item it points to a valid previous block, and
            \item the voter is legitimate, i.e., $VoterStake()$ returns the
                positive stake amount as declared.
        \end{inparaenum}
        After successful verification, votes are added to the pending
        list of votes that directly support its predecessor block.  These votes create a
        so-called \textit{virtual block} that consists of collected but not yet included votes, and one virtual block can support
        another virtual block; we discuss this further in \autoref{sec:details:forks}.

        After waiting for $\Delta$ seconds to collect and
        validate votes, nodes execute the round's second step.  First, the node checks the output of the $LeaderStake()$
        function, and if it is positive then the node is a
        \textit{leader} in that round.  If so, then the node determines the main chain -- see the details in
        \autoref{sec:details:forks}. The node then creates and propagates a new
        block that has the main chain's last block -- which can be a virtual block -- as its predecessor and which includes, among other fields (see \autoref{eq:block}), all
		collected votes and the generated random value $r_i$. A malicious leader can censor -- i.e., refuse to include -- votes, but we use our incentive mechanism, which is described in \autoref{sec:details:rewards}, to discourage this attack.
	 
        A node that receives a new block verifies whether
        \begin{inparaenum}[a)]
            \item it is authentic, formatted correctly, and not from the future,
            \item it points to a valid previous block,
            \item the votes are correct, and
            \item the leader is legitimate, i.e., $LeaderStake()$ returns a
                positive value.
        \end{inparaenum}
        If the block is validated, it is appended to its corresponding chain.
        Besides pointing to the previous block, a leader in its block lists all
        known forks that were not reported in previous blocks, including pending
        votes of other blocks.

We propose a concrete instantiation of voter/leader election in \autoref{sec:details:election}.
\name can also be implemented with other procedures (e.g., based on VRFs as discussed in \autoref{app:vrf}) as long as nodes act as leaders
and voters in proportion to their stake possession.   
We do not restrict the roles of a node per round, i.e., a node can both be a
voter and leader in a given round, or act in only one of these roles, or
none.

\subsection{Leader and Voter Election}
\label{sec:details:election}
We propose a method of electing leaders and voters which is based on
a novel \textit{cryptographic sampling} method presented in \autoref{alg:election}.  This method creates an array of all stake units and pseudorandomly
samples a fraction from it.  The method uses
 uniquely generated PRF outputs to sample stake units. 
 In a round, leader and voter elections should be independent, thus
the $Sample()$ function takes a role parameter -- `lead' and `vote', respectively -- to
randomize PRF outputs for these two elections.  The function returns a list of
sampled public keys -- each key corresponds to a stake unit -- and is parametrized
by the size of the output list. In the following, $q$ denotes how many
stake units out of the total stake $n$ are elected every round for the voting committee, and $l$ is
an analogous parameter for the number of leaders. The $VoterStake()$ and $LeaderStake()$ functions
run $Sample()$ and return how many times the given public key is present in the
sampled stake.  In App.~\ref{app:prf-sampling}, we show that our  construction is
indistinguishable from random sampling for computationally bounded
adversaries.  As stake units are sampled uniformly at random, a node
with stake $s$ can be elected as a voter between 0 and $s$ times in any given
round.  For performance reasons it may be desirable to elect one leader per round, which is achieved by setting $l=1$.

\begin{figure}
\removelatexerror
\begin{algorithm}[H]
\caption{The voting procedure.}
\label{alg:voting}
\footnotesize
\SetKwProg{func}{function}{}{}
\func{VotingRound(i)}{ 
    $r \leftarrow RoundBeacon(i)$;
    $s\leftarrow VoterStake(pk,r)$\; 
    \If(// check if I am a voter){$s>0$}{
        $B_{-1} \leftarrow MainChain().lastBlk$; // get last block\\
        $\sigma\leftarrow Sign_{sk}(i\|H(B_{-1})\|s)$\;
        $v\leftarrow (i,H(B_{-1}), s, pk, \sigma)$; // support vote\\
        Broadcast($v$)\;
    }
    Wait($\Delta$); // meanwhile, collect and verify support votes\\
    \If(// check if I am a leader){$LeaderStake(pk,r)>0$}{
        $B_{-1} \leftarrow MainChain().lastBlk$; // possibly different block\\
        $V\leftarrow\{v_a, v_b, v_c, ...\}$; // received $B_{-1}$'s support
     votes\\
        $r_i \leftarrow Random()$;
        $\sigma\leftarrow Sign_{sk}(i \| r_i\| H(B_{-1}) \|F\| V \| \mathit{Txs})$\;
        $B \leftarrow (i, r_i, H(B_{-1}),F, V, \mathit{Txs}, pk, \sigma)$; // new block\\
         Broadcast($B$)\;
     }
     Wait($\Delta$); // wait for the next round
}
\end{algorithm}
\begin{algorithm}[H]
    \caption{Leader/voter election via cryptographic sampling.}
\label{alg:election}
\footnotesize
\SetKwProg{func}{function}{}{}
\func{VoterStake(pk, r)}{
    $tmp\gets Sample(q, r, \textit{`vote'})$;
    \Return $tmp.Count(pk)$\;
}
\func{LeaderStake(pk, r)}{
    $tmp\gets Sample(l, r, \textit{`lead'})$;
    \Return $tmp.Count(pk)$\;
}
\func{Sample(size, r, role)}{
    $tmp \gets []$;
    $res\gets []$\;
    \For{$pk \in stake$}{
        \For{$s \in \{1, ..., stake[pk]\}$}{
            $tmp.Append(pk)$\;
        }
    }
    \For{$i \in \{1,...,size\}$}{
        $k\gets PRF_{r}(i\|role)\ \%\ Len(tmp)$\;
        $res.Append(tmp[k])$;
        $tmp.Delete(k)$\;
    }
    \Return $res$\;
}
\end{algorithm}
\end{figure}

\myparagraph{Limitations}
The described approach guarantees that in every round the exact same stake fraction is sampled.
As a result, nodes are able to more quickly make commit decisions.
However, a disadvantage is that an adversary may try to launch an adaptive
attack -- e.g., (D)DoS -- as elected nodes are known before they broadcast their
messages.  Fortunately, multiple lightweight mechanisms that provide
network anonymity are available. For instance,
Dandelion~\cite{bojja2017dandelion,fanti2018dandelion++}, which is intended for use in
cryptocurrencies, provides formal anonymity guarantees combined with low latency and
overheads.  Using such a mechanism together with \name would complicate attempts by the
adversary to identify the node's IP address, effectively mitigating the
mentioned (D)DoS attack.

Another
way of addressing this issue
in PoS blockchains is to elect nodes using cryptographic primitives with secret inputs (e.g., VRFs as in Algorand \cite{gilad2017algorand}, or unique signatures).  Using this approach, a node's role can
be revealed only by this node itself, e.g., by propagating a message.  The
disadvantage, as we show in \autoref{sec:algorand}, is the resulting variance in the number of elected entities, which slows the block commitment process.  In App.~\ref{app:vrf}, we
show how \name can be combined with VRF-based election and that our
commitment scheme is  still applicable.  An efficient mechanism that combines
``secret'' election with fixed committee sizes is an
open research problem.

\subsection{Chain Selection}
\label{sec:details:forks}
\name does not follow the longest-chain rule of Bitcoin's NC -- instead, the strength of a chain
is expressed by the stake that supports its blocks.  To improve the handling of
forks, we incorporate the GHOST protocol~\cite{sompolinsky2015secure}, which
improves the throughput and security of NC by utilizing fork blocks in the calculation of the total PoW
weight of chains. 

\myparagraph{Forks and Virtual Blocks}
In \name, votes contribute to ``weights'' of chains and are crucial to determine the main branch.  In short, they represent beliefs of
stakeholders about their views of the main chain.  To compare chains, nodes
follow the
\textit{most-stake} rule, i.e., the chain which is supported by more
stake-weighted votes is
chosen. Ties can be broken in \name by
 selecting the chain whose hash value computed over the concatenated round
beacon and whose last leader's public key is smaller. \name allows situations in which no block is added in a round -- e.g., when a faulty node is
elected as a leader or when the network is temporarily asynchronous -- or in which the block
 contains few or no votes.  In such cases, nodes create a virtual block
(see \autoref{sec:details:voting}) which is a set of received votes that have not yet
been included in the main chain. 

Virtual blocks do not have transactions or
signatures, unlike `standard' blocks.
During chain selection, nodes do not distinguish virtual from ``standard''
blocks: if a virtual block is stronger than a conflicting standard block, then nodes will support the virtual block. Voters can support a virtual block by voting for the block's latest standard ancestor (see \autoref{fig:virtual}). In later rounds, leaders can collate non-included votes per round to create a sequence of virtual blocks, of which the latest is used as the predecessor of the newly proposed standard block. The virtual blocks are then transmitted by the leader along with the standard block. A leader who neglects to include votes in her block
risks that it is overwritten by another leader who aggregates
the non-included votes in virtual blocks. The overwritten block may still be referenced by a later block using GHOST as discussed below, but even then the malicious leader still loses her block reward (see \autoref{sec:details:rewards}). If a vote legitimately appears in two conflicting blocks, e.g., in a virtual block and an overwritten but referenced standard block from the same round, then one is disregarded. Since virtual blocks are only created through standard blocks, the commitment procedure of \autoref{sec:details:final} is only executed on standard blocks. 

\begin{figure}[t!]
  \centering
  \includegraphics[width=.9\linewidth]{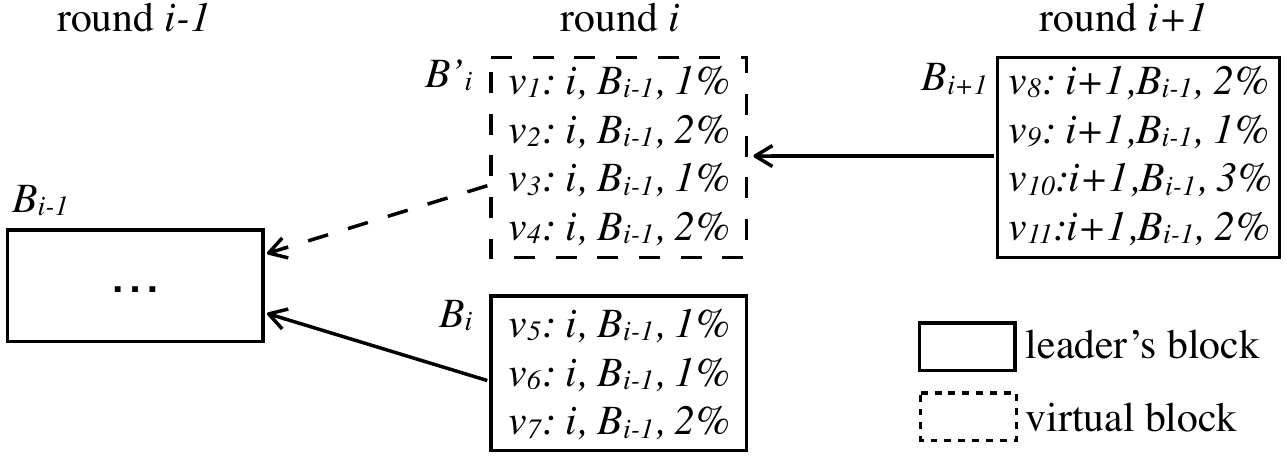}
    \caption{An example of a virtual block fork.}
  \label{fig:virtual}
\end{figure}

An example is presented in \autoref{fig:virtual}, where $q/n=10\%$ and voters in
round $i$ publish seven votes supporting the block $B_{i-1}$.
The leader of round $i$ creates a block $B_i$ which includes only three
votes with 4\% of the stake. Therefore, a virtual block $B'_i$ with 6\% stake is
created and in round \emph{$i+1$} all votes implicitly support this block
instead of $B_i$, thus the leader of this round creates a block $B_{i+1}$ that aggregates
votes of the virtual block and points to $B_{i-1}$ via $B'_{\;i}$.  We emphasize that the
block $B_{i+1}$ can also include a pointer to $B_i$ cf.\ GHOST.

\myparagraph{Subtree Selection (GHOST)}
The proposed protocol
is likely to work well when the network is highly synchronous -- which can be achieved by choosing a sufficiently high value for $\Delta$.  However, $\Delta$ must be traded off against transaction throughput: i.e., throughput can only be high if $\Delta$ is low. If $\Delta$ is small compared to the network latency, then
 asynchronous periods in which blocks cannot reach nodes
before the defined timeouts occur often. This would result in a high stale block
ratio that harms the security of system.

In order to prepare the protocol to withstand such situations, we modify and
extend GHOST to adapt it to our setting.  The chain selection procedure is
depicted in \autoref{alg:selection}.  As presented, the \textit{MainChain}
procedure starts at the genesis block.  Then for each of its child blocks, the
algorithm calculates the total stake in the child block's subtree,
and repeats this procedure for the child block with the most stake aggregated
on its subtree, and so on.  When the protocol terminates it outputs the block which
denotes the last block of the main branch.
The chain selection procedure relies only on  the stake encoded in votes and
collected votes of virtual blocks -- i.e., those not included in any actual block --
and includes them in the total stake of their chain.  

\begin{algorithm}[t!]
\caption{The chain selection procedure.}
\label{alg:selection}
\footnotesize
\SetKwProg{func}{function}{}{}

    \func{MainChain()}{
        $best \leftarrow (.lastBlk\leftarrow B_0, .stake\leftarrow Stake(B_0))$\;
    \While{$True$}{
        \If{$Children(best.lastBlk)=\emptyset$}{
            \Return $best$\;
        }
        \For{$B \in Children(best.lastBlk)$}{
            $s\leftarrow TreeStake(B)$\;
            \If{$s>best.stake$}{
                $best \leftarrow (B, s)$; // stronger subtree found\\
            }
        }
    }
}
\end{algorithm}

\begin{figure}[t!]
  \centering
  \includegraphics[width=\linewidth]{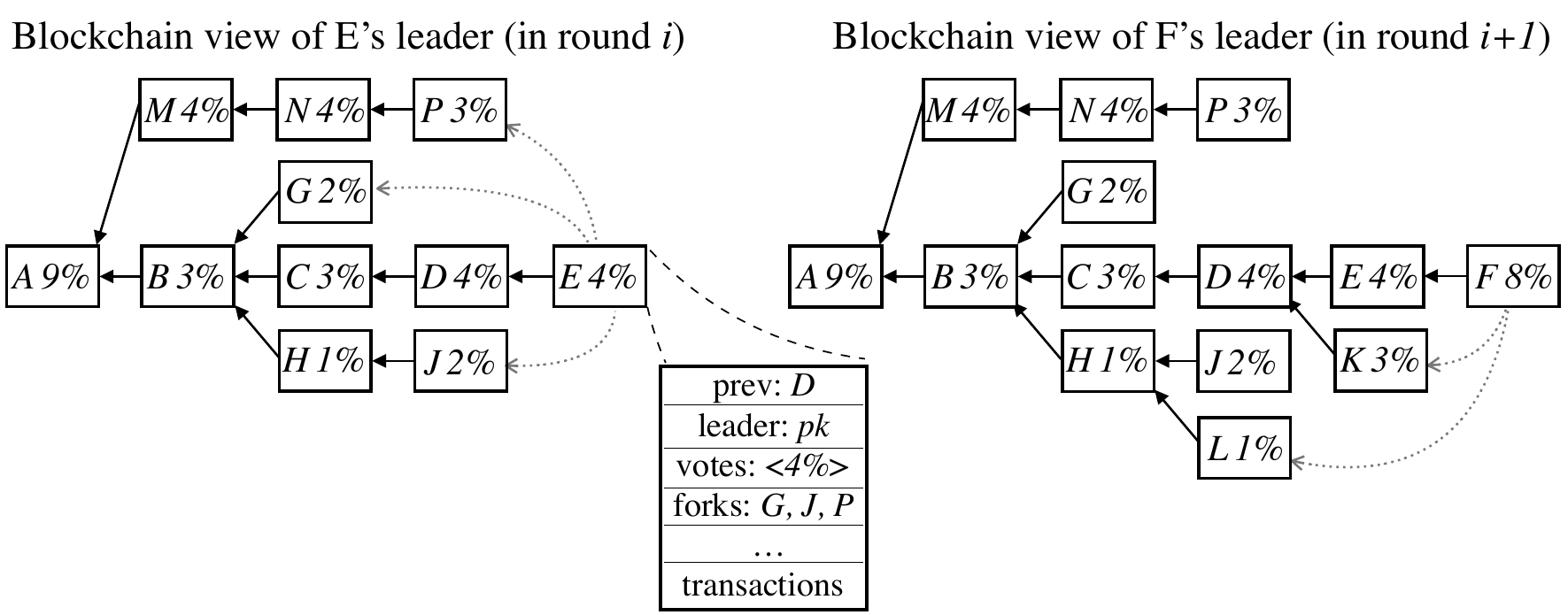}
    \caption{Example of the evolution of the blockchain when many forks occur. In this figure (and in Fig.\ \ref{fig:weak_preceding}~and~\ref{fig:prob}), the percentage in each block denotes the supporting stake included in that block.}
  \label{fig:chain}
\end{figure}

 To illustrate the chain selection procedure we show an example in \autoref{fig:chain} where $q/n=10\%$. In our example, in round $i$, the leader (i.e., the
creator of block $E$) sees four chains, i.e., those ending with the blocks $P$,
$G$, $D$, and $J$, respectively.  To determine the main chain, the leader runs
the protocol starting with the genesis block $A$ and computes the stake of its
children's subtrees. Block $M$'s subtree stake is $11\%$ whereas $B$'s subtree stake is $15\%$. Thus, $B$ is selected and the leader
repeats the same procedure for this block. Finally, the leader determines
$A,B,C,D$ as the main chain, and creates a new block $E$ pointing to it. 
Moreover, the leader introduces pointers (dashed lines) to known but yet unreported fork blocks ($P,G,J$).
These pointers are required to preserve the integrity of the blockchain view.
Note that the chain
$A,B,C,D$ is selected as the main one, despite the fact that the stake that
voted for this chain is lower than the stake of the chain $A,M,N,P$. As
different chains within a subtree support the same chain prefix, an
advantage of combining GHOST with the most-stake rule is that it requires an
adversary to compete with the entire subtree and not only its main chain,
which makes attacks on safety much more difficult.  Repeating the procedure, round
$i+1$'s leader extends $E$'s chain and reports on fork blocks $K$
and $L$.

\subsection{Block Commitment}
\label{sec:details:final}
Block commitment in \name is decided by each node individually.  Each node
specifies $p^*$, its risk level, and $\commitb$, the block to be committed. Given its
current view of the blockchain, it then calculates the probability that the target block
$\commitb$ can be reverted.  Given our chain selection procedure, this question
can be reformulated as follows: \textit{what is the probability that an
adversary can create any stronger subtree than the corresponding subtree
containing $\commitb$?} If the probability is less than the threshold $p^*$,
then the block is committed.

An
adversarial subtree has to be rooted at a block outside $\commitb$'s
{supporting subtree}, as
otherwise, it would support $\commitb$.
The stake of the supporting subtree is known to the node and computed simply by
$s=TreeStake(\commitb)$.
An adversarial subtree can originate from any previous
round, so we require that the block $\commitb$ cannot be committed before all of its previous blocks have also been committed. For each block, we hence consider the potential strength of an adversarial subtree that originates in the parent of $\commitb$ and which has $k$ supporting blocks until the current round. 
As such, the supporting and adversarial branches are in competition during $k$ rounds. 
The node cannot determine
how many stake units support an adversarial subtree, but this knowledge is critical for the
security of the commit operation. Therefore, the node splits this stake into
the sum of
\begin{inparaenum}[a)]
    \item \textit{missing stake} which consists of those stake units that may
        unknowingly contribute to the adversarial branch during a fork, and
    \item \textit{adversarial stake} which is the sum of the stake that the
        adversary accumulated over the $k$ rounds.
\end{inparaenum}

The adversarial stake units can \textit{equivocate}, which means that they are involved in the creation of two or more different blocks or votes within the same round. 
As such, the stake units that contribute to the adversarial branch can also be present on $\commitb$'s supporting subtree. Furthermore, their exact quantity is unknown, although on average this will be equal to $\alpha q$. 
We discuss equivocations in more detail below. We assume the worst-case scenario regarding the
missing stake, i.e., that all missing stake -- even if it is honest -- is placed in an
unknown adversarial subtree.  That could happen during asynchronous periods or
network splits. The node also conservatively assumes that the adversary would win
every tie.  However, we emphasize that an adversarial subtree must
be internally correct, e.g., it cannot have equivocating votes within it, as otherwise, honest nodes would not accept it.

\begin{algorithm}[t!]
\caption{The block commitment procedure.}
\label{alg:commit}
\label{alg:probcomphyper}
\label{alg:probboundcramer}
\footnotesize
\SetKwProg{func}{function}{}{}
\func{Commit(lastCommit,$p^*$)}{
	$B \leftarrow \textit{GetMainChainBlockAtRound}(lastCommit)$;
	
    \While{$B\neq nil$}{
        $s\leftarrow TreeStake(B)$;
        $k\gets CurrRound()-lastCommit+1$\;
        \If(// $B$ is rejected){$\mathit{CalculateProb(k, s)} \geq p^*$}{
            \Return;
        }
        $lastCommit \leftarrow lastCommit + 1$\;
        $B \leftarrow \textit{GetMainChainBlockAtRound}(lastCommit)$\;
    }
}
\func{CalculateProb($k$,$t$)}{
	\If{$q^{k}$ $<$ MAX\_STAKE}{ \label{ln:computationthres}
		\Return $HyperGeomProb$($n$,$\frac{1}{2}(1+\alpha)n$,$q$,$k$,$t$);
	}
		\Return $CCBound$($n$,$\frac{1}{2}(1+\alpha)n$,$q$,$k$,$t$);
}
\func(// hypergeom. sum prob.){HyperGeomProb($n$,$u$,$q$,$k$,$t$)}{ 
    \If{$k = 1$\label{ln:cdfstart}}{ 
		$x \leftarrow 0$;
		
        \For{$i \in \{t+1,\ldots,q\}$}{
          $x \leftarrow x + \textit{HypergeometricPMF}(n,u,q,i)$;
        }
        \Return $x$; \label{ln:cdfend}
    }
    $x \leftarrow 0$;
    
    \For{$i \in \{0,\ldots,q\}$}{
    	$x \leftarrow x + \textit{HypergeometricPMF}(n,u,q,i)$ \\ \hskip0.6cm $*$ $\textit{HyperGeomProb}(n, u, q, k - 1, s - i)$;
       }
    \Return $x$;
}
\func(// the Cram\'er-Chernoff bound){CCBound($n$,$u$,$q$,$k$,$t$)}{ 
    \Return $e^{-k * \textit{MaximizeRateFunc}(n,u,q,t/k)}$;
}
\func{MaximizeRateFunc($n$,$u$,$q$,$t$)}{
	$\lambda_{\max} \leftarrow \textit{RateFuncSearchRange}(n,u,q,t)$;
	
	\Repeat{$\Delta < \epsilon$} {
		
		\For{$i \in \{1,\ldots,4\}$}{
			$y_i \gets \textit{RateFuncHelper}(n,u,q,\lambda_i,t)$;
		}
		
		$\vec{\lambda} \leftarrow \textit{ReduceInterval}(\vec{\lambda},\vec{y},n,u,q,t)$;
		
		$\Delta \leftarrow \lambda_{i_{\max}} - \lambda_0$;
		
		// $\epsilon$ is an accuracy threshold set by the user
	}
	\Return $(y_{i_{\max}} - y_0)/2$;
}
\func{RateFuncHelper($n$,$u$,$q$,$\lambda$,$t$)}{
	\Return $t\lambda - \textit{LogHyp2F1}(-q,-u,-n,1 - e^{\lambda})$;
}
\end{algorithm}

\myparagraph{Hypothesis Testing}
To commit $B$, the node conducts a
statistical \emph{hypothesis test} to assert whether most of the network has the
same view of the supporting tree as the node. To do so, the node computes a so-called \textit{$p$-value} that represents the probability of achieving a total supporting stake $s$ over $k$ rounds \textit{under the hypothesis} that fewer stake units are contributing to the supporting branch than to the adversarial branch. We commit $B$ if this $p$-value is so low that we can safely conclude that this hypothesis is invalid. The function \textit{Commit} in \autoref{alg:commit} is called once every round, and as many blocks as possible are committed during every call. 
To achieve safety, we use the quorum intersection argument present
in traditional BFT systems, except in our case it is probabilistic. Namely, an
adversary in such a split network could produce equivocating votes on both views at the
same time; however, the node can conclude that such a situation is
statistically unlikely if she sees that her target block is supported by more
than $\frac{2}{3}kq$ stake (see details in \autoref{sec:analysis:safety}). That
would imply that most of the honest nodes see the target block on the main chain,
as any alternative adversarial view cannot obtain more than $\frac{2}{3}kq$
over time (since $\alpha<\frac{1}{3}$).

The process is described by the pseudocode of \autoref{alg:commit}.  At the core
of the procedure, the node keeps committing the main chain's subsequent blocks by computing
the probability that their corresponding supporting trees are on the ``safe'' side of
the hypothetical fork.  Depending on the parameters, the $p$-value is computed
using one of two functions, namely $HyperGeomProb$ or $CramerBound$ (see
\autoref{alg:probcomphyper}). In \autoref{sec:analysis:safety}, we discuss these
algorithms in more detail and show that they indeed give an upper bound on the error
probability.

To keep our presentation simple, we do not parametrize the commit procedure by
the $f$ (or similarly $\alpha$) parameter. However, it would not change our methodology,
and in practice, it may be an interesting feature as nodes could then base commitment
decisions on their own adversarial assumptions in addition to their security level $p^*$.

\myparagraph{Block Commitment Example}
\begin{figure}[t!]
  \centering
  \includegraphics[width=.695\linewidth]{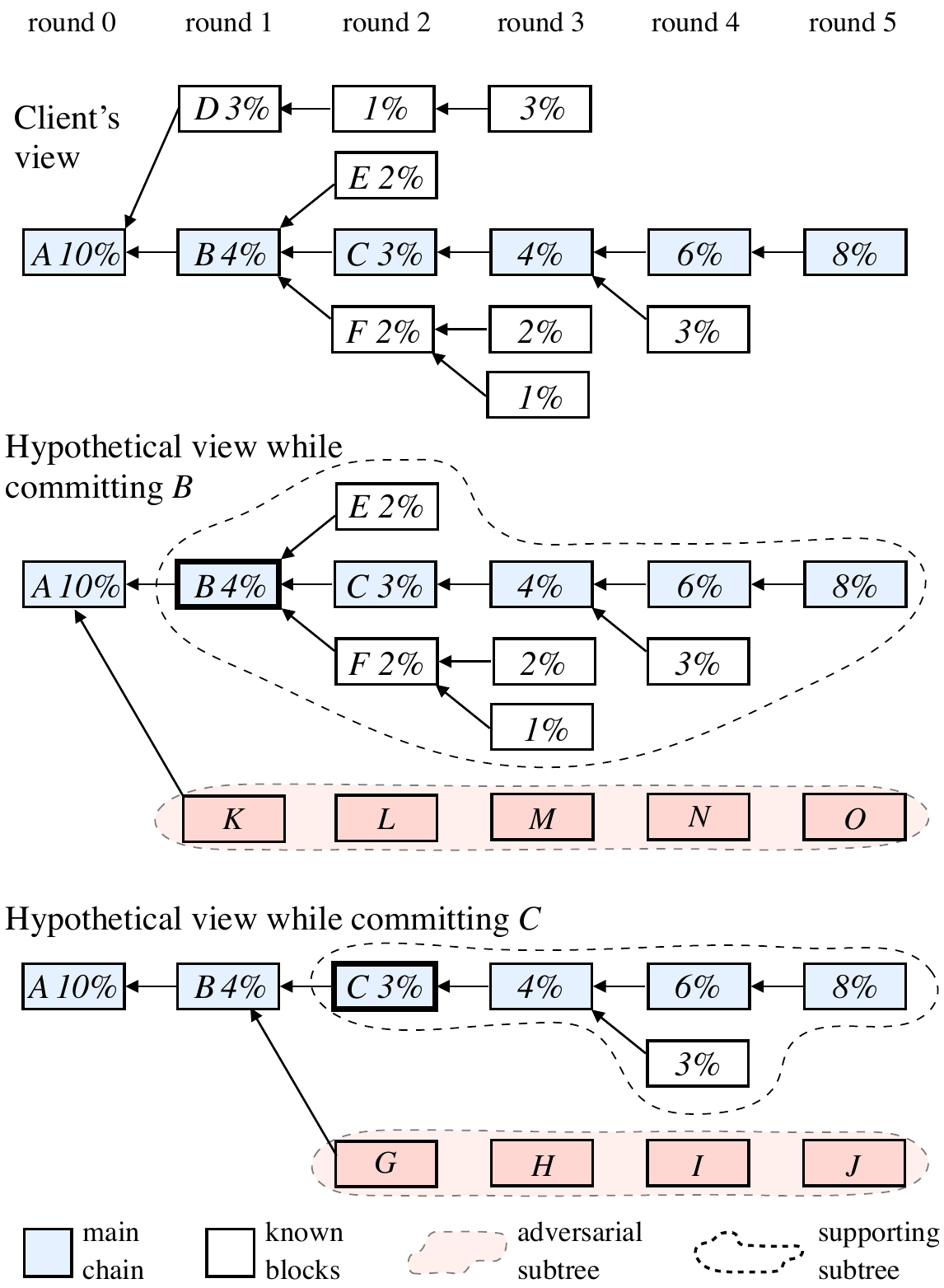}
    \caption{An example of the client's view of the blockchain and possible
    adversarial subtrees. (Percentage of a block denotes the stake directly
    supporting this block.)}
  \label{fig:prob}
\end{figure}
To illustrate the block commitment process better we present a blockchain example in
\autoref{fig:prob} where it holds that $q = \frac{1}{10} n$ and a client wishes to commit the block $C$ which is on the
main chain (denoted in blue).  To overwrite the block, an adversary would need
to overwrite $C$'s or $B$'s supporting subtree. However, since block $B$ must have already been committed, we do not need to consider any subtree starting from $A$ (although we show it in the figure). As described, the
client computes the $p$-value for the hypothesis that more stake units contribute to the adversarial subtree than to the supporting subtree under the worst-case conditions, and commits if this is low enough.

In our case, if the client wishes to commit $C$, then the adversarial
subtree that could overwrite the supporting tree directly originates
from block $B$ and has $k=4$ potential blocks, i.e., $G,H,I,J$. (We do not consider the
topologies of adversarial subtrees, as only the stake that they aggregate
decides their strengths.)  The supporting subtree has
$24\%$ of the stake. The missing stake
in this case is $4\cdot 10\% - 24\% = 16\%$. However, the exact amount of missing stake that contributes to an adversarial chain may depend on the random beacon implementation (e.g., block $J$ and the main chain block in the same round may have different committees if the beacon depends on some information in the blocks from the preceding rounds). Hence, we do not use this information directly. 
To be able to conclude with some degree of certainty that enough stake units are supporting $B$, we need the supporting stake to be comfortably above
$\frac{2}{3}kq$.  
If the calculated $p$-value is acceptable for the client, the client
commits $B$.

\myparagraph{Equivocation}
Although adversarial nodes can freely violate the rules of the protocol, in \name,
messages -- i.e., blocks and votes -- are signed, so adversarial nodes can be held accountable for
certain equivocations. In particular, the following actions are
detectable:
\begin{inparaenum}[a)]
    \item equivocating by producing conflicting votes or blocks within the same round, or
    \item supporting or extending a chain that is weaker than a previously
        supported or extended chain.
\end{inparaenum}
The former is a nothing-at-stake attack and it is provable by
showing adversary's two signed messages which support or extend different chains
in the same round.
In the latter case,
the adversary violates the chain selection rule. This can be proven 
by any pair of signed messages $(m_1, m_2)$ which support/extend two different
chains $C_1$ and $C_2$, respectively, where: 
$Round(m_1) < Round(m_2) \wedge Stake(C_1)>Stake(C_2)$.

Although prevention of such misbehavior is challenging, solutions that
disincentivize it by causing the misbehaving validator to lose all or parts of
her deposit/stake have been proposed~\cite{buterin2017casper}. Under this approach, the protocol allows honest nodes to submit evidence of equivocation to the blockchain for a reward, e.g., the finder's fee implemented in the smart contract version of Casper FFG \cite{buterin2019incentives}. Implementing such a scheme in \name is an interesting direction for future work.

\subsection{Random Beacon}
\label{sec:details:beacon}
\name relies on a pseudorandom beacon to elect round leaders and voters.  It is
important for security that these beacons are difficult to be biased by adversaries. 
The safety analysis in \autoref{sec:analysis:safety} relies on the assumption that voters and block proposers on each branch of a fork are sampled proportionally to their stake -- if the adversary is able to manipulate the random beacon's committee selection, then this assumption no longer holds. Furthermore, a beacon that is too predictable is vulnerable to adaptive network-level attacks (e.g., DoS) against voters and round leaders, as the window for such attacks is directly proportional to the time between the creation of a block and the publication of the random beacon.

In this paper we do not propose a new random beacon
construction -- instead, we rely on previously proposed
concepts for a concrete instantiation.  For the current implementation, we have followed the approach by Daian et
al.~\cite{daian2017snow}, where beacons are generated purely from the random
values aggregated over ``stable'' main chain blocks.  More precisely, as
presented in \autoref{sec:details:voting}, an $i$-th round leader inserts a
random value $r_i$ into its block.  For the security parameter $\kappa$, which is the
number of main chain blocks after which the probability of a roll back is
negligible, the random beacon $r$ in the current round $j$ is extracted in
a two-step process, using
a random oracle, from the
previous random values
$r_{j-2\kappa},r_{j-2\kappa+1},r_{j-2\kappa+2},...,r_{j-\kappa}.$
An adversary can bias the outcome beacon $r$, but it has been
proven~\cite{daian2017snow} that short-time adversarial biases
are insufficient to get a long-term significant advantage. 

We emphasize that the use of a random beacon for leader/committee selection is
not specific to \name and other
random beacon constructions can also be used to implement the protocol. Some recent approaches, e.g., DFINITY \cite{hanke2018dfinity} and RandHound~\cite{syta2017scalable}, promise scalable
randomness and are potentially more bias-resistant than the presented mechanism. However, greater bias resistance may come at the cost of additional computational overhead.
We leave the investigation of these schemes combined with \name as future work.

\subsection{Rewards}
\label{sec:details:rewards}
\name introduces the rewarding scheme presented in \autoref{alg:rewards}.
Each voter supporting the previous block of the main chain receives a voter
reward $R_v$ multiplied by the number of stake units that the voter was sampled
with.  The votes of virtual blocks 
also receive rewards as soon as a block
containing them is published.  Every leader who publishes a block on the main chain
receives a leader reward $R_l$ and
an inclusion reward $R_i$ for every stake unit included in the block.
Leaders may also receive transaction fees paid by nodes but we omit them as
they are application-specific.

\begin{algorithm}[t!]
\caption{The reward procedure.}
\label{alg:rewards}
\footnotesize
\SetKwProg{func}{function}{}{}
\func{RoundReward(B)}{
    $pay(B.pk, R_l)$; // leader's block reward (+ opt. tx fees)\\
    \For{$v \in B.V$}{
        $pay(v.pk, v.s*R_v)$; // voter's reward\\
        $pay(B.pk, v.s*R_i)$; // leader's inclusion reward\\
    }
}
\end{algorithm}

The rewarding scheme in \name has several goals.  First, it aims
to
incentivize voters to publish votes supporting their strongest views
immediately.  Forked or ``late'' votes and blocks are marked in main
chain blocks but they are not rewarded although leaders still have an incentive to include
them since they strengthen common ancestor blocks with the main chain.
Therefore, voters trying to wait for a few blocks to vote for ``past best blocks''
would always lose their rewards.  Second, it incentivizes leaders to publish
their blocks on time, as a block received after the end of the round would not be part
of the main chain. After all, voters in the next round would follow their strongest
view instead, so the leader would miss out on her rewards.  Lastly, the
scheme incentivizes block leaders to include all received votes. A leader
censoring votes loses an inclusion reward proportional to the stake that
was censored. Moreover, the censoring leader weakens her own blocks.

\section{Analysis}
\label{sec:analysis}
\subsection{Probabilistic Safety}
\label{sec:analysis:safety}

In this section we show that a block $B$ is committed only if the probability that a conflicting block is committed is indeed below $p^*$. To do this, we use a novel proof technique based on \emph{statistical hypothesis testing}. A block is committed by a user when she sees enough supporting stake to conclude that it is unlikely that she is on a branch that includes only a minority of honest users. The main threat is an adversary who wants to cause a safety fault, which means that two conflicting blocks are committed by different honest users. For this to occur, two users who have \emph{different} views of the blockchain must \emph{both} see enough evidence for their blocks to commit. 

Our worst-case scenario -- i.e., the scenario in which a safety fault is most likely to occur -- is when the honest users are split during a {fork}. The adversary can make a safety fault more likely by voting on both branches of the fork simultaneously. Although equivocation can be punished in retrospect, it cannot be detected while the fork is ongoing and hence cannot be ruled out for our safety analysis.
We assume that a fraction $\advfrac \in [0,\frac{1}{3})$ of the stake is controlled by the adversary. We furthermore assume that the fork consists of \emph{two} branches, and that the honest stake is split \emph{evenly} among them. The reason is that if one branch is stronger than the other, then it is less likely that a block will be committed on the weaker branch. Similarly, there is no need to consider forks with three or more branches, because the probability of conflicting blocks being committed would be higher if two of the branches were combined. Finally, we do not consider users going offline or votes being withheld by the adversary -- this would only make one branch weaker and a safety fault less likely, and therefore constitute a weaker adversarial scenario than the one under consideration. 

The user cannot directly observe how many users are on the two branches -- however, she \emph{can} observe how many stake units support the block on her branch. The expected stake fraction on her branch in the worst-case scenario is $\advfrac + \frac{1}{2}(1-\advfrac) = \frac{1}{2}(1 + \advfrac)$, so $\frac{2}{3}$ if $\advfrac=\frac{1}{3}$ (this is the same for the users on the other branch). If she observes that the amount of stake on her branch is considerably higher than this fraction, then she has evidence to conclude that her branch is stronger. In the following, we will formulate the question of whether there is sufficient evidence to commit $B$ as a {hypothesis test}. Before we proceed, we note that for $B$ to be committed, we require that all preceding blocks have also been committed. Hence, we only focus on $B$ and its supporting subtree, and do not consider any of $B$'s ancestors (see also the example in \autoref{fig:weak_preceding}). 

Recall that there are $n \in \mathbb{N}$ stake units of equal weight, where a single node may control many stake units. We draw a committee of size $q$ stake units in every round (e.g., in \autoref{fig:weak_preceding} it holds that $q = \frac{1}{10} n$).  We assume that $n$ and $q$ remain constant throughout the duration of the fork, as these are protocol-level parameters that change infrequently. We denote the number of expected supporting stake units per round on the user's branch by $u$. We also assume that $u$ does not change during the fork -- i.e., that users do not move between branches while the fork is ongoing. We make this assumption because 1) we assume that the adversary is unable to adaptively move users between branches (as stated in \autoref{sec:pre}), and 2) if any of the honest users becomes aware of the other branch then they would be able to detect equivocation and send evidence to users on both branches. 
We assume that in a given round $l$, the user is interested in testing whether the amount of stake that supports a block $B$ that appeared in round $j$ is enough to commit it. The \emph{null hypothesis} $H_0$ asserts that the supporting stake is \emph{at most} equal to the expected worst-case support, i.e.,
\begin{equation}
H_0: u \leq \frac{1}{2}(1 + \advfrac) n. 
\label{eq:nullhypothesis}
\end{equation}
In the following, we compute the probability $p$ of observing a given amount of supporting stake in $B$'s subtree \emph{given} that $H_0$ is true. This probability is commonly called a $p$-value. If the $p$-value is below $p^*$, then we accept the alternative hypothesis $H_1 : u > \frac{1}{2}(1 + \advfrac) n$ and commit the block.
 Of course, this experiment may be repeated over the course of several slots until $B$ is finally committed, or if $B$ is dropped after the resolution of the fork. Hence, we use a \emph{sequential} hypothesis test. In the following, we first focus on calculating the probability of observing the data -- i.e., the supporting stake for the block $B$ -- given the null hypothesis for given rounds $\cround$ and $j$, and then extend this to the sequential setting.

\begin{figure}[t!]
\centering
  \includegraphics[width=.7\linewidth]{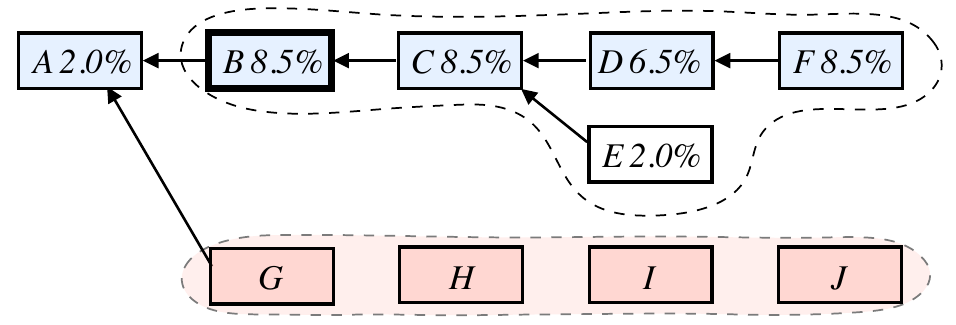}
\caption{If the user tries to commit $B$, we only evaluate the strength of $B$ and its supporting subtree, and implicitly compare it to the branch $(G,H,I,J)$. Even though $A$ is weak, and $B$ might be overturned due to a branch starting from $A$'s predecessor, we require that $A$ has already been committed before  committing $B$, so it need not be considered.}
  \label{fig:weak_preceding}
\end{figure}

Let $X_i$ be the number of stake units held by supporting users sampled in round $i$. Since the cryptographic sampling algorithm draws without replacement, we know that $X_i$ has a \emph{hypergeometric} distribution \cite{johnson2005univariate}, whose probability mass function is given by
$$
\mathbb{P} (X_i = x) = \frac{\comb{u}{x} \comb{n - u}{q-x}}{\comb{n}{q}}.
$$
The \emph{total} number of supporting units (where a single unit may feature in more than one different rounds) that is drawn in rounds \mbox{$j+1,\ldots, \cround$} is then given by
$$
T = \sum_{i = j+1}^{\cround} X_i.
$$
Since $X_1, X_2, \ldots$ are all identically distributed, we know that $T$ has the distribution of $k = \cround - j$ independent hypergeometric random variables.
Hence, to calculate the probability of observing a total amount of supporting stake $t$ given the null hypothesis, we can compute the tail probability $\P(T \geq t) = \sum_{i=t}^{kq} \P(T = i)$ under the assumption that $u = \frac{1}{2}(1+\alpha)n$. After all, of all possible values for $u$ included under $H_0$, this choice would give the highest probability.
Unfortunately, a sum/convolution of hypergeometric random variables does not have a probability mass function that is easy to compute. It can be expressed as 
$$
\P(T = t) = \sum_{x_2=0}^q \sum_{x_3=0}^q \ldots \sum_{x_{k}=0}^q \prod_{i=2}^{k} \P(X_i = x_i) \P\left(X_{1} = t - \sum_{i=2}^{k} x_i\right).
$$
This expression is derived in Lemma~\ref{lm:convolutions} in the appendix. The intuition behind it is as follows: for all possible sequences $(x_2,x_3,\ldots,x_m)$, we calculate the probability of observing it, and then multiply this by the probability that $X_1$ is exactly $t$ minus the sum of the sequence. This is what is implemented in \autoref{alg:probcomphyper}: the recursion encapsulates the repeated sum, while the final computation of $\P(X_1 = t - \sum_{i=2}^k x_i)$ is done in lines~\ref{ln:cdfstart}~through~\ref{ln:cdfend}. Numerically, it can be simplified somewhat -- e.g., it is not necessary to consider sequences whose sum already exceeds $t$. However, to compute this probability it is inescapable that the number of operations is roughly proportional to $q^k$, which is very large even for moderate values of $k$. 
This is the intuition behind the if-statement in line~\ref{ln:computationthres} of \autoref{alg:commit}: the exact computation is only performed if $q^{m-j} = q^k$ is small enough, i.e., below some threshold that depends on the node's processing power.

One less computationally expensive approach would to find an approximation for the distribution of $T$. Good candidates include a single hypergeometric (with parameters $kq$, $ku$, and $kq$ instead of $n$, $u$, and $q$), a binomial (for which we draw \textit{with} instead of \textit{without} replacement), a Poisson (which approximates the binomial distribution), or a normally distributed (by the Central Limit Theorem) random variable. However, in the following we instead focus on establishing bounds using the Cram\'er-Chernoff method \cite{boucheron2013concentration}. The reasons for this are twofold. First, it gives us a strict upper bound on the probability of interest regardless of the parameter choice, which is safer than an asymptotically valid approximation. Second, it can be more easily generalized to settings where the distribution of the random variables is slightly different. We have also investigated other methods for bounding the tail probabilities of $T$, e.g., the Chernoff-Hoeffding bound, and the approaches of \cite{klenke2010stochastic} and \cite{teerapabolarn2015binomial}. However, we found that the Cram\'er-Chernoff bounds were the sharpest while still being computationally feasible.

The Cram\'er-Chernoff method can be used to find an upper bound on the probability $\P(X \geq x)$ for a random variable $X$. Its basis is Markov's inequality \cite{boucheron2013concentration}, which states that for any nonnegative random variable $X$ and $x > 0$,
$$
\P(X \geq x) \leq \E(X)/x.
$$
It can then be shown \cite{boucheron2013concentration} that for any $\lambda \geq 0$,
$$
\P(X \geq x) = \P(e^{\lambda X} \geq e^{\lambda x}) \leq \E(e^{\lambda X}) e^{-\lambda x}  = e^{-(\lambda x - c_X(\lambda))},
$$
where $c_X(\lambda) = \log(\E(e^{\lambda X}))$.
Since this holds for all $\lambda \geq 0$, we can choose $\lambda$ such that the bound is sharpest. Let 
$$
r_{X}(x) = \sup_{\lambda \geq 0} (\lambda x - c_{X}(\lambda)).
$$
It then holds that 
$$
\P(X \geq x) \leq e^{-r_X(x)}.
$$
In particular, if \mbox{$T = X_1 + \ldots + X_k$}, such that all $X_i$, $i \in \{1,\ldots,k\}$ are mutually independent and have the same probability distribution as $X$, then
$$
c_T(\lambda) = \log(\E(e^{\lambda T})) = \log\left(\prod _{i=1}^k \E(e^{\lambda X_i})\right) = k c_{X}(\lambda)
$$
and hence
$$
r_T(k x) = \sup_{\lambda \geq 0} (\lambda k x - k c_{X}(\lambda)) = k r_X(x),
$$
and therefore
$
\P(T \geq t) \leq e^{-k r_X(t/k)}.
$
Common names for $r_{X}$ include the \textit{Legendre-Fenchel transform}\footnote{Strictly speaking, only after broadening the range of the supremum.} 
or the large-deviations \emph{rate function} of $X$ after its use in Cram\'er's theorem for large deviations \cite{boucheron2013concentration,dembo1998large}. In our context, $t/k$ is the average supporting stake accumulated per round. As shown in Lemma~\ref{lm:meansmalllft} in the appendix, if $t/k$ is higher than the expected worst-case supporting stake, which in our case is $q u/n$, then the probability that the adversary wins decreases exponentially in $k$. In fact, the function $r_X$ represents the exponential rate at which this probability decays (hence  the name ``rate function''). 

As an example, let $n=1500$, $u=1000$, and $q=150$. We find that $r_{X}(\lfloor \frac{3}{4} q \rfloor) = r_{X}(112) \approx 2.50$. This means that the probability under the null hypothesis that we observe an average support of $75\%$ of the committee for $k$ rounds in a row decreases at least exponentially with a factor of $e^{-2.50} \approx 0.082$ per round. This is displayed in \autoref{fig:probgraph}, for $\bar{s} = t/q = 75\%$. After 15 rounds, the probability of observing support from, on average, $75\%$ of the committee is below $10^{-16}$ under $H_0$.

The function $c_{X}(\lambda)$ does not have a convenient closed-form expression: it is equal to $_2F_1(-q,-u,-n,1-e^{\lambda})$ \cite{zwillinger2000crc}, where $_2F_1$ is the Gaussian or ordinary hypergeometric function.\footnote{\url{https://en.wikipedia.org/wiki/Hypergeometric\_function}} 
Hence, the rate function has to be evaluated numerically. Since $c_{X}(\lambda)$ is convex on $\lambda \geq 0$ (see also Lemma~\ref{lm:convexcgf} in the appendix), there is only one local optimum and we can therefore use some variation of
golden-section search \footnote{\url{https://en.wikipedia.org/wiki/Golden-section\_search}} 
to find it. The procedure in \autoref{alg:probboundcramer} is as follows: we first determine a search range for $\lambda$ by widening a base interval until we observe that the function decreases at some point in the interval. We then iteratively narrow the interval until its width is below some accuracy threshold. 
A baseline implementation is given in \autoref{alg:helpers} in the appendix. 
Computation of the rate function is still not trivial, although it depends on the size of $q$: for $q=150$, we have found that roughly $15$ instances can be computed per second on a MacBook Pro with a 2.5 GHz Intel Core i7 processor, compared to $280/s$ for $q=30$ and $0.6/s$ for $q=750$. Still, although there is no theoretical limit on the number of blocks that can be committed in a single round via the \textit{Commit} function in \autoref{alg:commit}, in practice it is unlikely that more than several blocks are committed in a single round, as the evidence for committing is accumulated slowly. Furthermore, evaluations of $r_X(x)$ can be cached to speed up the \textit{Commit} function in later rounds.

\begin{figure}[ht!]
\centering
    \begin{subfigure}[t]{.24\textwidth}
      \centering
      \includegraphics[width=\linewidth]{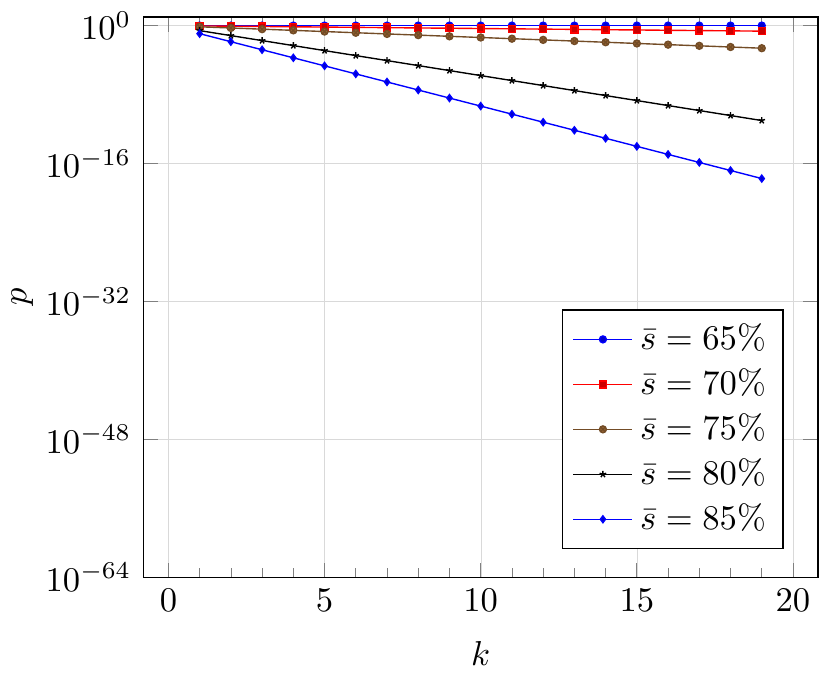}
      \caption{$q=30$}
      \label{fig:sub-first}
    \end{subfigure}
    \begin{subfigure}[t]{.24\textwidth}
      \centering
      \includegraphics[width=\linewidth]{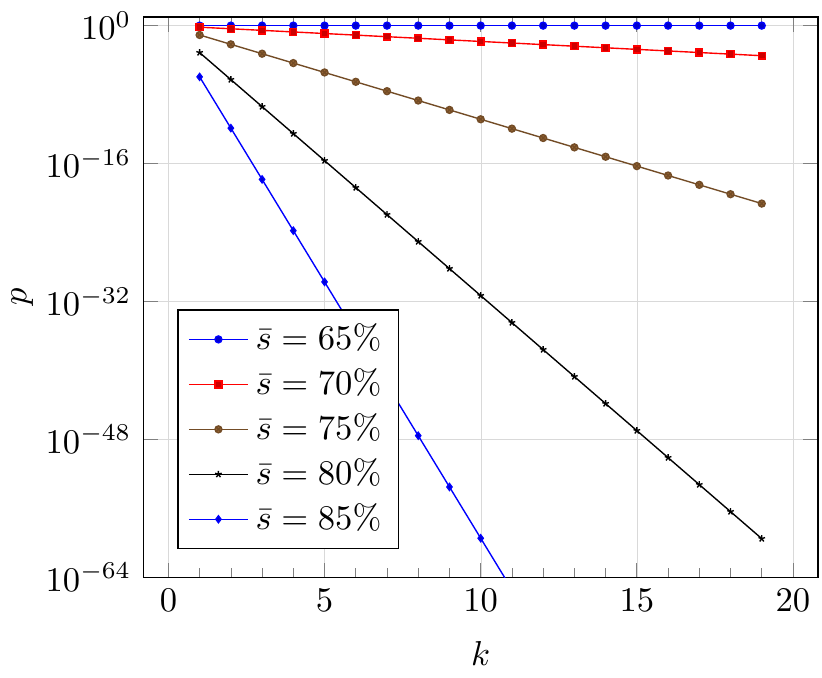}
      \caption{$q=150$}
      \label{fig:sub-second} 
    \end{subfigure}
    \begin{subfigure}[t]{.24\textwidth}
      \centering
      \includegraphics[width=\linewidth]{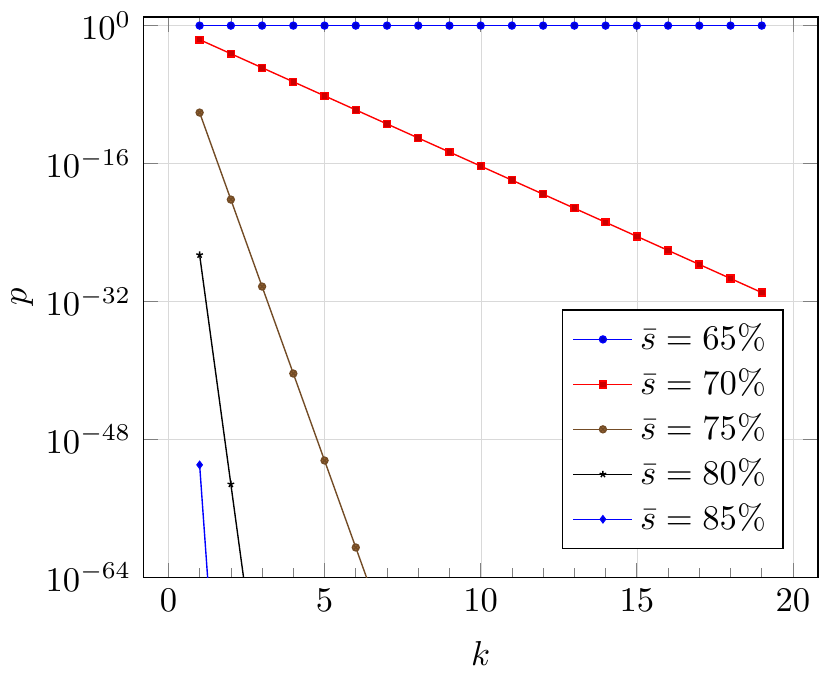}
      \caption{$q=750$}
      \label{fig:sub-third}
    \end{subfigure}
\caption{The Cram\'er-Chernoff bound $p$ for the probability that a block $B$ is overturned, assuming that $B$ and its supporting subtree in the $k-1$ subsequent rounds contain the votes of a fraction $\bar{s}$ of the committee. In these figures, $n = 1500$, $u = 1000$, and $q$ varies. Clearly, when the committee size is larger, the $p$-values vanish faster and nodes are able to commit sooner. }
  \label{fig:probgraph}
\end{figure}

\autoref{fig:probgraph} depicts the evolution of the bounds if the same average supporting stake $\bar{s}$ is observed over $k$ rounds consecutively. It is clear that, as expected, larger committee sizes mean that large deviations are less likely, and therefore that the user is able to commit sooner. Hence, there is a trade-off between security and efficiency. Even if $q=30$ and on average 24 (i.e., $\bar{s} = 80\%$) stake units vote to support a block, then the $p$-value decreases by almost $75\%$ per round (i.e., $e^{-r_X(24)} \approx 0.263$). 

Although we have so far focused on a single test conducted in round $m$ about a block $B$ in slot $j$, in practice the test for $B$ will be performed multiple times if there is not enough evidence to commit immediately. Each time the test is performed, there is some probability of error, so we need to account for this accumulation. The most straightforward way to achieve this is as follows: if we conduct the test for the $i$th time, we use $p^* \cdot \gamma^i$ for some $\gamma \in (0,1)$ as our threshold. The total probability of error after a potentially infinite number of trials is then bounded by $\sum_{i=1}^{\infty}  p^* \gamma^i = \frac{\gamma}{1-\gamma} p^*$. For example, if $\gamma = \frac{9}{10}$, then we need to multiply our probability threshold by $\frac{1-\gamma}{\gamma} = \frac{1}{9}$ (e.g., if we originally had set a threshold of $10^{-16}$ then we would now use $\approx$$1.11 \cdot 10^{-17}$), and the threshold for the $p$-value would become lower (so harder to meet) by $10\%$ each round. This is not a major obstacle: as mentioned before, even for a committee size of $30$ and an average supporting stake of $80\%$, the $p$-value decreases considerably faster than by $10\%$ per round, and the multiplicative factor of $1/9$ is met after only 2 additional rounds.

Using the above, we can prove the following theorem. \begin{theorem}
    \label{theorem:liveness}
\name achieves probabilistic safety. 
\end{theorem}
\begin{proof}
A safety fault occurs when two honest users on conflicting branches commit a block on their branch. If the honest stake is split evenly between the two branches, then the probability that both users commit is $(p^*)^2$, which is much lower than $p^*$.\footnote{It can be even lower if there is negative correlation between the stake on the two branches, e.g., if the implementation of the random beacon is such that the same committee is chosen for different blocks in the same round.} In the case where the stake is \textit{not} evenly split among the two branches, then if the user is on the weaker branch, the probability that she ever commits $B$ is below $p^*$. If the user is instead on the stronger branch, then the probability that the honest users on the weaker branch ever commit is below $p^*$. Hence, the probability of both $B$ and a conflicting block being committed is below $p^*$ in all settings. Our safety property is therefore satisfied.
\end{proof}

\myparagraph{Long-Range Attacks}
An additional benefit of the commitment scheme in \name is its potential to act as a first line of defense against \textit{long-range attacks}. A long-range attack occurs when the adversary obtains a set of keys that control a large majority of stake in the distant past. Unlike PoW, blocks in PoS protocols are not computationally hard to create, which means that it is feasible for an adversary to create an extremely long branch by herself (e.g., $10^{6}$ blocks or more). 
If the adversary forks the chain in the far past and creates a branch that is stronger than the main branch, then the attacker can convince honest nodes to revert a large number of their blocks in favor of adversarial blocks. We assume that the system's users will eventually be able to agree that the original branch is the ``correct'' branch after a period of out-of-band consultation. However, in the meantime honest nodes are vulnerable to double-spend attacks as they may accept transactions that conflict with transactions on the eventual main chain.

A user can mitigate long-range attacks by instructing her client to \textit{finalize} committed blocks, i.e., to not revert them without a manual reset. However, if the node commits a block on the losing branch of a temporary fork, then it will remain out-of-sync with the rest of network until the user intervenes. This again leaves the node vulnerable to double-spend attacks, and its votes and blocks will not earn rewards on the main chain. It is therefore only profitable to use the above defense mechanism if its expected gains outweigh its costs. Although not all quantities involved in these costs and gains are clear a priori, the user \textit{does} have a bound on the probability of committing a block on the losing branch of a fork, namely $p^*$ as discussed above. In the following, let $C_1$, $p_1$, and $C_2$ be the user's guesses for the expected total cost of a long-range attack, the per-block probability of a long-range attack,\footnote{To illustrate this: if the user expects long-range attacks to occur once per year, and if there are $10^6$ blocks per year, then $p_1 = 10^{-6}$.} and the expected total cost of being out-of-sync with the network until a manual reset, respectively. Then the gains of this defense mechanism (i.e., nullifying the cost of long-range attacks) outweigh its costs if 
$C_1 p_1 > C_2 p^*$, i.e., if $p^* < \frac{C_1}{C_2} p_1.$

In practice, $C_1$ and $C_2$ will depend strongly on the nature of the user. For example, as witnessed by the recent $51\%$ attacks on Bitcoin Gold
and Ethereum Classic,\footnote{See also the online references \href{http://bitcoinist.com/51-percent-attack-hackers-steals-18-million-bitcoin-gold-btg-tokens}{here} and \href{https://www.coindesk.com/ethereum-classic-suffers-second-51-attack-in-a-week}{here}.} cryptocurrency exchanges are a valuable target for double-spend attacks -- after all, exchanges routinely process high-value transactions and their actions (e.g., transactions on another blockchain) are often irreversible. These nodes can reflect this by choosing a high value for $C_1$. By contrast, a user who controls stake units will miss out on voting and block creation rewards when she is out-of-sync with the network, so she will judge $C_2$ to be much higher than $C_1$ if she is unlikely to be targeted by attackers. 

In general, if the user deems long-range attacks to not be likely or costly, or if she deems the costs of being out-of-sync to be high, then she will judge the defense mechanism to be worth it only if $p^*$ is low. As can be seen from Figure~\ref{fig:probgraph}, the $p$-value of our test decreases rapidly over time, so even if $p^*$ is extremely low, then blocks will be committed in a timely manner -- e.g., for $q=30$ and $24$ supporting stake per block on average, $p^*$ reaches $2^{-256} \approx 10^{-77.06}$ after only 133 rounds. Furthermore, the user is free to set a lower threshold for finalization than for regular commitment decisions, and compare the calculated $p$-value for each block to both thresholds. Hence, she can set the threshold for finalization very low, and still be safe from long-range attacks from fewer than $10^3$ blocks in the past at negligible expected cost.

\subsection{Liveness}
In our analysis we assume the partially synchronous network model, where after
the GST the network is synchronous for $t$ rounds.
We follow the same setup in \autoref{sec:analysis:fairness} and
\autoref{sec:analysis:rewards}.
To prove liveness we use two lemmas (see App.~\ref{app:liveness}).  In short,
they state that with increasing $t$, after $t$ rounds, a)  the main chain
includes at least $t(1-\alpha)q$ stake units (Lemma~\ref{lemma:strongchain}), and b)
the main chain includes $m$ blocks that the honest participants voted for, where
$t\geq m\geq \lceil t(1-2\alpha)\rceil$ (Lemma~\ref{lemma:nchain}). As stated in \autoref{sec:pre}, honest nodes are allowed to be temporarily offline before but not after the GST, and permanently offline honest nodes are grouped with the adversarial nodes.

\begin{theorem}
    \label{theorem:liveness}
\name achieves liveness. 
\end{theorem}
\begin{proof}
    We show that the probability that within $t=2t'$ rounds an honest node
    appends a block to the main chain and the block is committed by every honest
    node is overwhelming, with $t$ increasing.

    By using Lemma~\ref{lemma:ngrowth}, we obtain that the $m'$ -- which increases
    proportionally with $t'$ -- blocks of the main chain were supported by, and thus
    visible to, the honest majority.  Now, to violate liveness, the adversary has to
    be elected as a leader for all of those contributed $m'$ rounds.\footnote{In
    such a case the adversary could append blocks with no transactions, append
    them to a different (weaker) fork, or not publish them at all.} The
    adversary is elected as a round leader with the probability proportional to
    her stake possession $\alpha<\frac{1}{3}$.  Therefore, the probability that at least one
    honest node adds a block in those $m'$ visible rounds is $1-\alpha^{m'}$, which
    with increasing $t'$ goes to $1$.

    Now, after the first $t'$ rounds, the block will be on the main chain, thus honest
    nodes over the next $t'$ rounds will be supporting the block by their
    stake of the total amount at least $t'(1-\alpha) q$ units (see
    Lemma~\ref{lemma:strongchain}) with $t'$ increasing. 
    We assume that $\alpha < \frac{1}{3}$, which means that there exists some $\epsilon>0$ such that \mbox{$\alpha = \frac{1}{3} - \epsilon$}. 
    This means that the expected amount of supporting stake per round will be $(\frac{2}{3} + \epsilon) q$, even if the adversarial nodes never vote. By the weak law of large numbers, the probability that the average supporting stake after $t'$ rounds is not above $\frac{2}{3}Q$ goes to zero as $t'$ goes to infinity. We know by Lemma~\ref{lm:meansmalllft} that if the observed average supporting stake is greater than the mean, the rate function evaluated at this average is greater than 0. By substituting the threshold $\alpha = \frac{1}{3}$ into \eqref{eq:nullhypothesis}, we find that $u \leq \frac{2}{3}n$, and that the expected per-round supporting stake under the null hypothesis $q u/n$ equals at most $\frac{2}{3}q$. Since we therefore observe more than the mean under the null hypothesis on average, the per-round decay rate will be greater than 1, and the upper bound on the test $p$-value goes to $0$ as $t'$ goes to infinity. As such, there will be some $t'$ at which the $p$-value goes below $p^*$, which means that the user can commit. This completes the proof for liveness.
\end{proof}

In practice, the length of $t$ depends on the values of $n$, $q$, $\alpha$,
and $p^*$, and can be derived from the bound given in
\autoref{sec:analysis:safety}.

\subsection{Fairness}
\label{sec:analysis:fairness}
We define $R(t)$ as the total expected reward of an honest node with $\beta$
stake fraction during $t$ synchronous rounds after GST.
In the following, we analyze the fairness of \name.
\begin{theorem}
    \label{theorem:fairness}
    $R(t)=\Theta(t)$
\end{theorem}

\begin{proof}
    Every round (see \autoref{alg:rewards}), a node has the expected voter's
    reward $e_v=\beta qR_v,$ and the expected leader's reward of
    $e_l=\beta(R_l + qR_i)$ (contributions of
    inclusion rewards $R_i$ may vary and in the worst
    case $e_l$ can be $\beta R_l$).  As shown in
    Lemma~\ref{lemma:ngrowth}, honest nodes participate in $n$ blocks of the
    main chain during $t$ rounds, with $n$ growing proportional to $t$.  We
    assume that an adversary can censor the node's votes (removing her voter
    rewards) but that will not happen in $n(1-\alpha)$ rounds on average (i.e., when
    an honest leader is elected).  Thus 
    $n(e_l + e_v)\geq R(t)\geq n(e_l + e_v(1-\alpha)),$
    and given the bounds for $n$ (Lemma~\ref{lemma:ngrowth}), the
    following holds

    {\hfill $t(e_l+e_v)\geq R(t) \geq \lceil t(1-2\alpha)\rceil(e_l+e_v(1-\alpha)).$}
\end{proof}

Another advantage of our reward scheme is that it rewards active nodes
frequently.  In every round, $q$ stake units are rewarded for voting and one
leader obtains the leader reward.  The rewards are given uniformly at random and
proportional to stake possession. In an example setting where
$n=10,000$, $q=200$, and $l=1$, even a node with a small stake possession
$\beta=1$ would receive a voter reward every 50 rounds and a leader reward
every $10^4$ rounds, on average.  With a realistic round time $2\Delta=5s$,
this would be 250 seconds and 13.9 hours, respectively.  By contrast, systems
that award only leaders would not give these frequent voter rewards.  For
instance, in Bitcoin (where only leader awards exist and rounds last 10
minutes on average) the node would receive a reward every 70 days on average.
By combining frequent rewards with fairness (\autoref{theorem:fairness}), \name
minimizes the reward variance which is seen as the root cause of creating large
mining pools which introduce centralization into the system.

\subsection{Rewards and Incentives}
\label{sec:analysis:rewards}
To reason about rewards in \name, we use a definition similar to the one by
Pass and Shi in~\cite{pass2017fruitchains}. We call a protocol's reward scheme
\textit{$\beta$-coalition-safe} if, with overwhelming probability, no $\beta' <
\beta$ fraction coalition can gain more than a multiplicative factor (1 +
$\epsilon$) of the total system rewards.  Intuitively, it means that with
overwhelming probability, for any coalition of nodes, the fraction of their
rewards  (within a time period) is upper-bounded by $(1+\rho)R(t)$, while a solo
node is guaranteed to receive at least $(1-\rho)R(t)$. Thus, the multiplicative
increase in the reward is
$\frac{1+\rho}{1-\rho} \leq 1+\epsilon$.

Using this definition, we show how relationships between rewards in \name
influence its coalition-safety.  As shown in  \autoref{theorem:fairness}, the
total reward of a node with $\beta$ stake fraction during $t$ synchronous
rounds is between $n(e_l + e_v)$ and  $n(e_l + e_v(1-\alpha))$.  Thus
$\frac{1+\rho}{1-\rho} = \frac{n(e_l + e_v)}{n(e_l + e_v(1-\alpha))}$
and the following holds: $\rho = \frac{\alpha e_v}{2e_l+e_v(2-\alpha)}$. We also have
$\frac{n(e_l + e_v)}{n(e_l + e_v(1-\alpha))} \leq (1+\epsilon)$ owing to
$\frac{1+\rho}{1-\rho} \leq (1+\epsilon)$, resulting in $\epsilon \leq
\frac{\alpha e_v}{e_l+e_v(1-\alpha)}.$

As we want  to minimize $\rho$ and $\epsilon$, owing to $e_v=\alpha q R_v,$ and
$e_l=\alpha(R_l + q R_i)$, we wish to minimize $\frac{\alpha q R_v}{2(R_l +
q R_i)+q R_ve_v(2-\alpha)}$ and $\frac{\alpha q R_v}{R_l + q R_i+q R_ve_v(1-\alpha)}$,
thus, $R_l$ should be large enough to make $\rho$ and $\epsilon$ smaller, i.e.,
$R_l\gg R_v$ and $ R_l\gg R_i$. 

We suggest that $R_v$ is higher than $R_i$. The intuition behind this choice is
that as voter rewards $R_v$ are received frequently when compared with inclusion
rewards $R_i$ received only by a leader.  A too high $R_i$ would contribute to a
high reward variance and could incentivize nodes to join their stake into pools
(as in PoW systems), introducing centralization risks.  Therefore, with
$R_v>R_i$, we propose the following relationships between rewards in \name:
\quad $R_l \gg R_v > R_i.$

Our rewarding scheme (\autoref{sec:details:rewards}), aims to incentivize
\begin{inparaenum}[a)]
    \item voters to broadcast their votes immediately,
    \item leaders to broadcast their blocks immediately, and
    \item leaders to include all received votes in their blocks.
\end{inparaenum}
We informally reason about its incentive-compatibility in a setting with
synchrony, a single leader elected per round, and rational nodes (optimizing
only for their rewards).  If a voter does not publish its vote on time, it will
not reach the round's leader, who subsequently cannot include this vote in her
block, thus the voter loses her reward $R_v$.  Similarly, if the leader does not
publish her block immediately, it is not received by other nodes, who in the
next round will vote for a virtual block, ignoring the late block and hence not awarding the rewards $R_v$ and $R_i$
to the leader.  If the leader publishes
blocks but without some of the received votes, then the leader is losing the inclusion
reward $R_i$ proportional to the stake of the ignored votes. 

\subsection{Scalability and Parametrization}
\label{sec:analysis:sim}
In order to evaluate how different configurations influence \name's performance
and security, we have built a \name simulator and conducted a series of experiments.
For our simulations, we introduced 5000 active nodes with one stake unit each
(thus $n=5000$) and a single leader per round ($l=1$). The number of nodes 
is
comparable with the estimated number of all active Ethereum nodes~\cite{ethnodes}.  To model the network
latencies while propagating messages we used the data from
\url{https://ethstats.net/}, which in real-time collects various information
(like the block propagation time) from volunteer Ethereum nodes.  In this
setting, we ran simulations for different values of $q$ and $\Delta$, where
each simulation run was for 10000 blocks.
We measured the block and vote stale rates $\psi_b$ and $\psi_v$, which respectively describe how
many blocks and votes do not make it onto the main chain.
\begin{figure}[t!]
    \centering
    \begin{subfigure}{.24\textwidth}
      \centering
      \includegraphics[width=\linewidth]{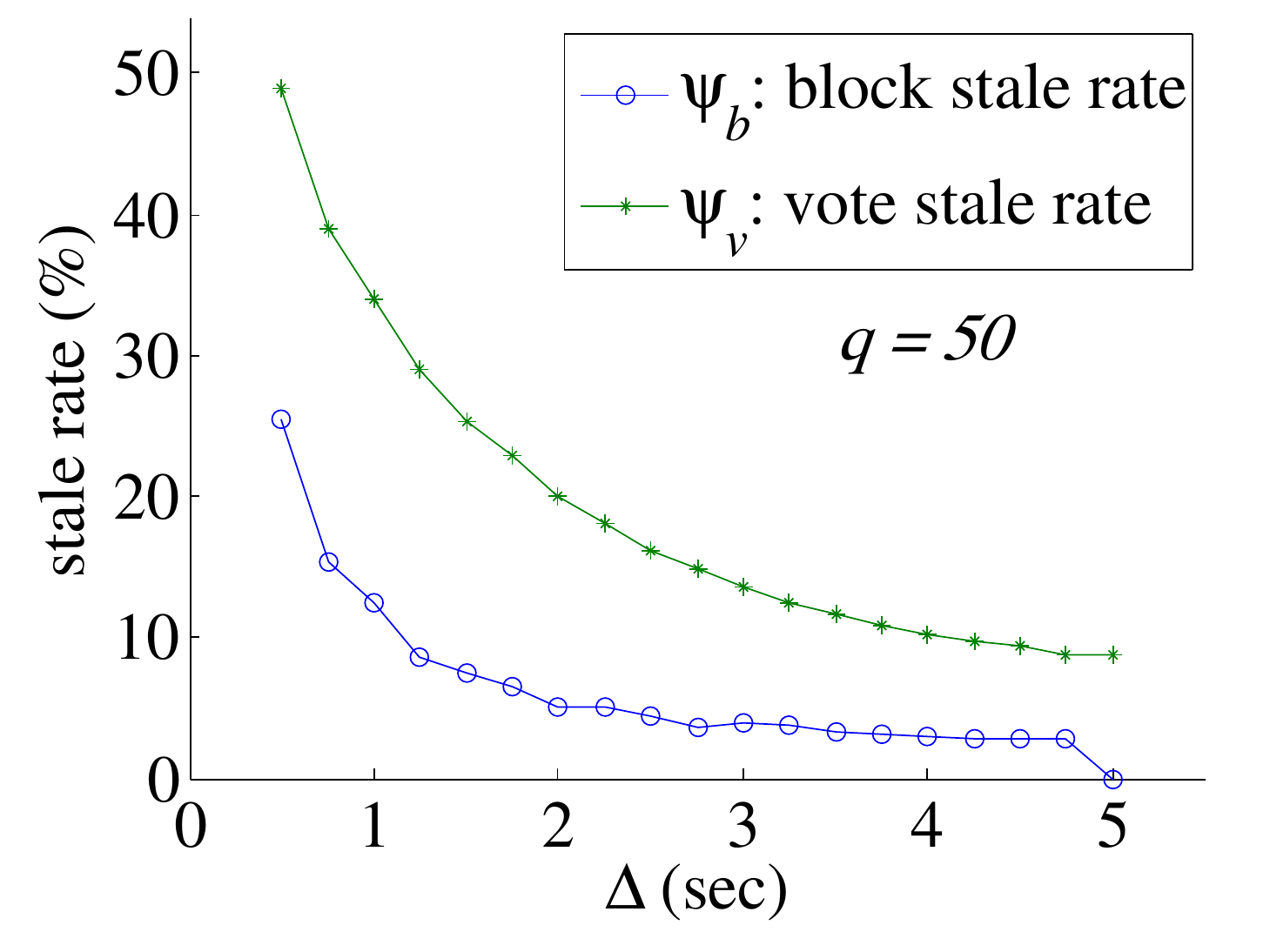}  
      \label{fig:sub-first}
    \end{subfigure}
    \begin{subfigure}{.24\textwidth}
      \centering
      \includegraphics[width=\linewidth]{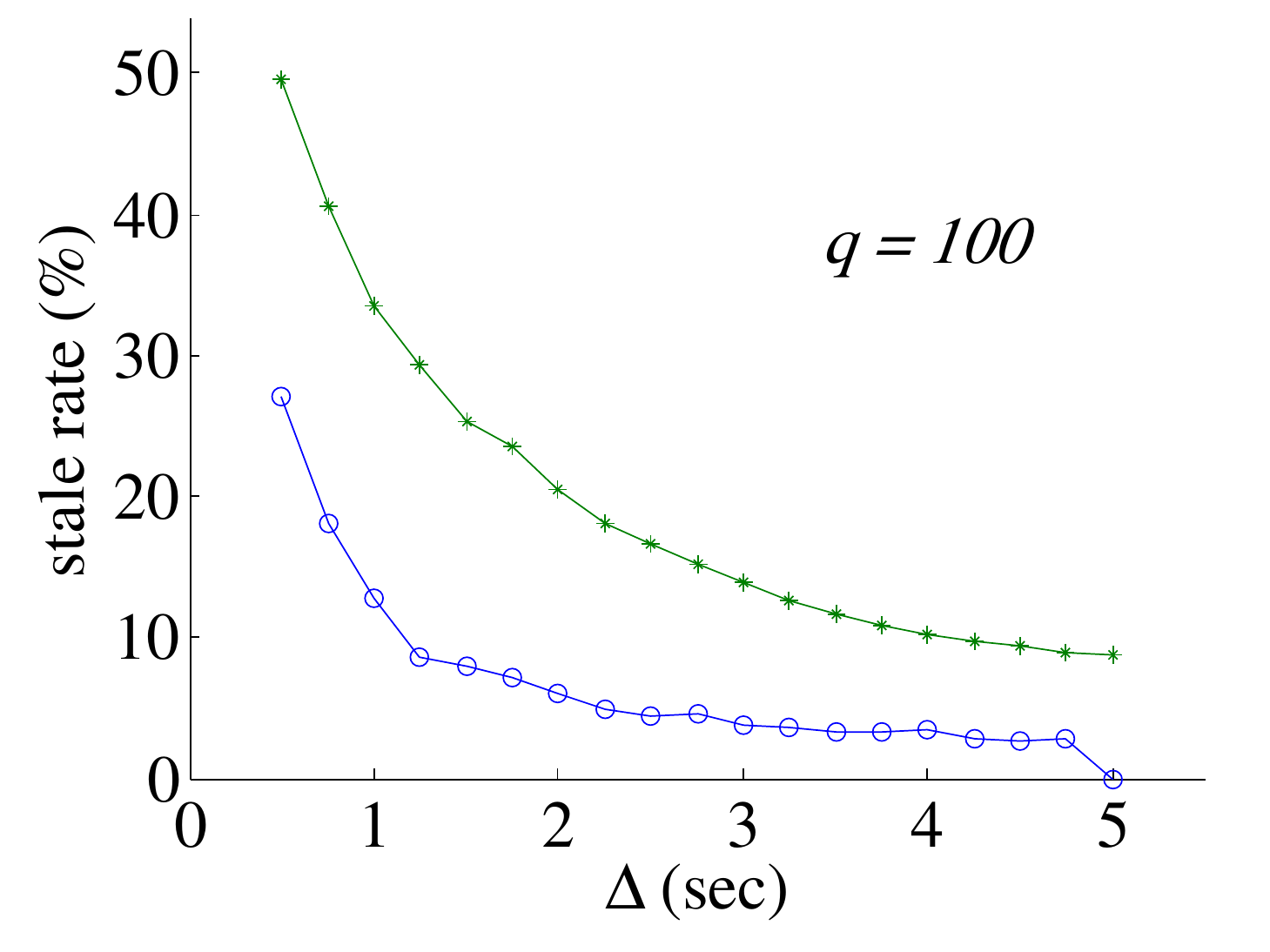}  
      \label{fig:sub-second}
    \end{subfigure}
    \begin{subfigure}{.24\textwidth}
      \centering
      \includegraphics[width=\linewidth]{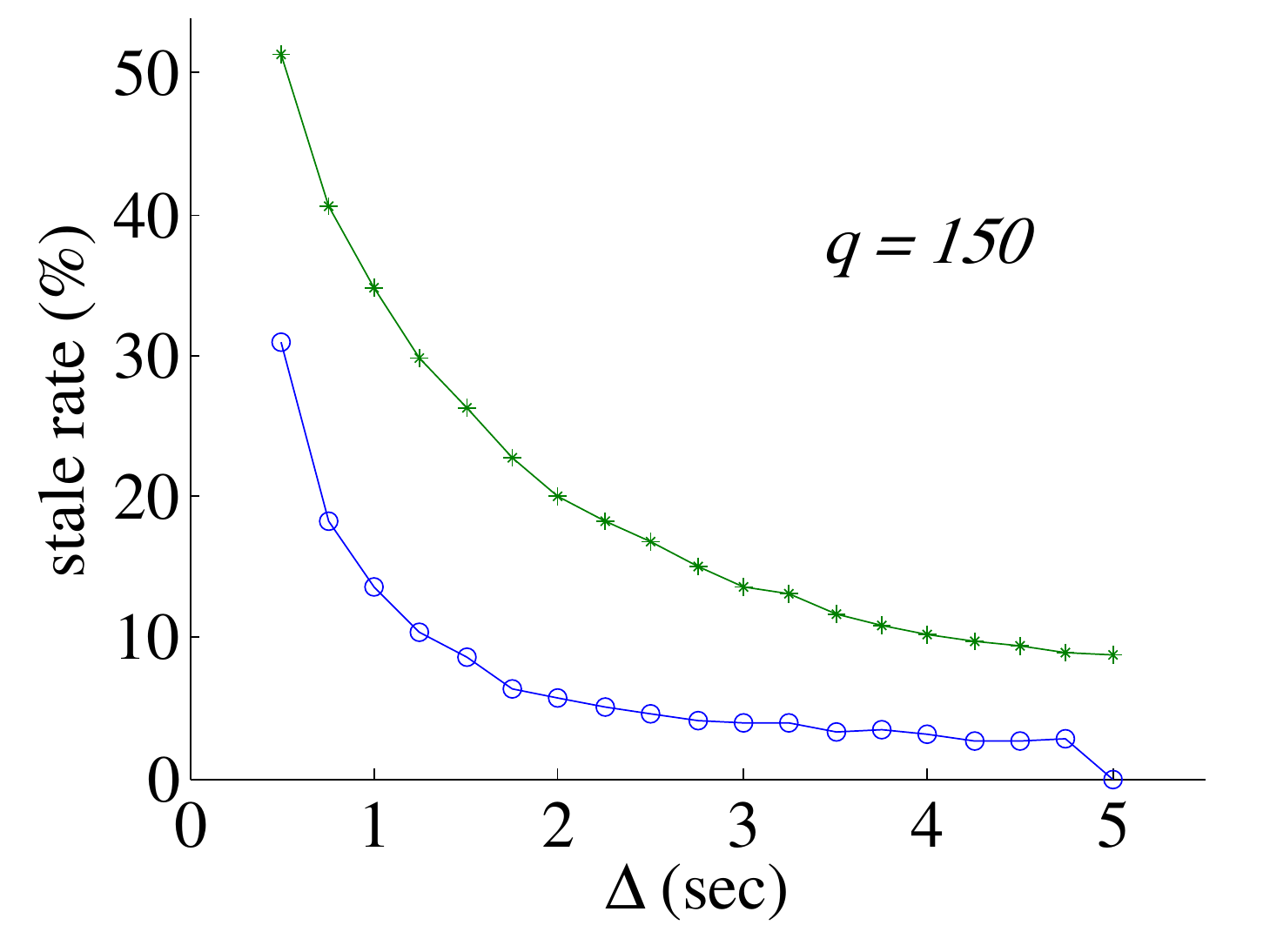}  
      \label{fig:sub-third}
    \end{subfigure}
    \begin{subfigure}{.24\textwidth}
      \centering
      \includegraphics[width=\linewidth]{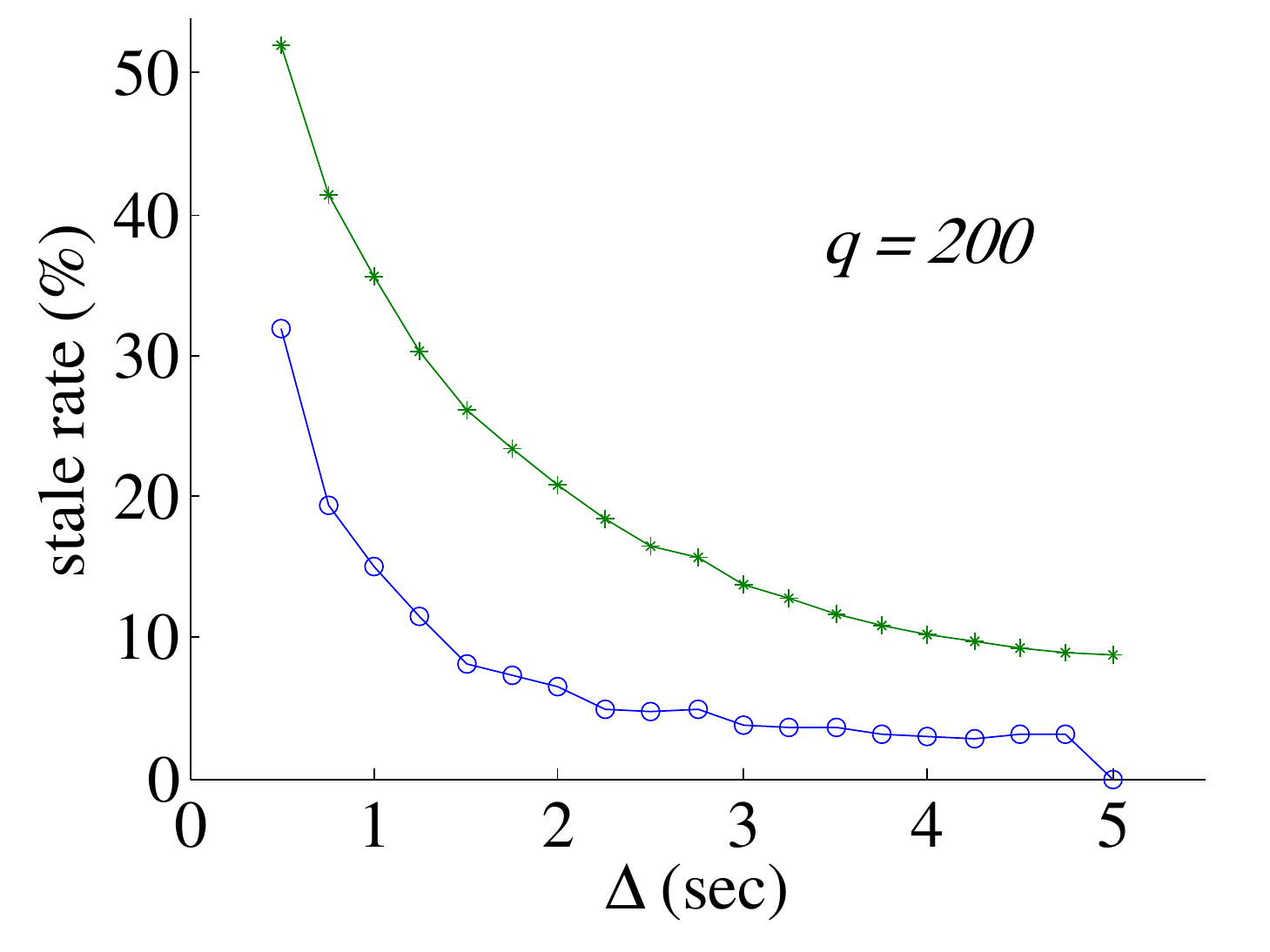}  
      \label{fig:sub-fourth}
    \end{subfigure}
    \caption{Simulated stale rates.}
    \label{fig:sim}
\end{figure}

The obtained results are presented in \autoref{fig:sim}. We see that if
$\Delta$ is too short, then protocol performance decreases as stale rates for both
blocks and votes are high. Besides influencing performance (fewer blocks
denote lower throughput), high stale rates degrade security as it takes
longer to commit a block (see \autoref{sec:analysis:safety}).  However, when $\Delta$ is slightly increased, the stale rates improve significantly.  We
found that  choosing $\Delta$ between 3-4 seconds may be good for the considered network
conditions.  We confirm our results in \autoref{sec:impl:eval}.  Moreover, with
increasing $q$, we see only a mild increase
in the stale rates.

\section{Implementation and Evaluation}
\label{sec:implementation}
\myparagraph{Implementation}
To test and evaluate \name we fully implemented it.
Our implementation is based mainly on Python.
For implementing the p2p networking stack, we use the Twisted framework while we
deployed Ed25519 signatures~\cite{bernstein2012high}.
The advantage of doing so is that is produces short public keys and signatures
(32 and 64 bytes, respectively).  In our implementation, we have encoded a vote in 80
bytes, and votes are included within their respective blocks.  Although the
number of votes is linear, even for large values of $q$ the introduced bandwidth
overhead should be acceptable. For instance, 1000 support
votes with our encoding would consume 80kB, which constitutes only 8\% of a 1MB
block or 4\% of a 2MB block.

\myparagraph{Evaluation}
\label{sec:impl:eval}
Equipped with our implementation we conducted a series of experiments in a
real-world setting. We set up a testbed consisting of physical machines
geographically distributed among 15 different locations from five continents. We
used these machines to run 100 \name nodes in total.  We used a simple
p2p flooding where every node peered with up to five
peers. We first investigated the
throughput of the network and we obtained an average end-to-end throughput of
around 10 Mbps. To introduce a conservative setting and better express
real-world heterogeneous network conditions we did not optimize our network by
techniques like efficient message dissemination or geographical peer selection.

In our setting, all our 100 nodes are elected as voters (i.e., $q=100$)
while only one node per round is elected as leader.  From the performance point
of view, such a setting would be for instance equivalent to the setting when
there are 2000 nodes in the system and $q/n=5\%$. During the execution of our
experiments, we noticed an imbalance between the durations of the
protocol steps. Namely, votes constitute only a small part of the blocks
that actually carry transactions, so if the nodes spend too much time waiting for votes then this would
not allow us to saturate the network. To maximize the throughput we
introduced separate waiting times for votes and blocks -- i.e., $\Delta_1$ and
$\Delta_2$ respectively, where $\Delta_1<\Delta_2$, instead of a single value $\Delta$ for both.

We have run a series of experiments with different $\Delta_1$, $\Delta_2$, and
block size
parameters and our results are presented in \autoref{tab:perf}.  Namely, we
measure the following three performance indicators: the block and vote stale rates
$\psi_b$ and $\psi_v$ as defined in \autoref{sec:analysis:sim}, and the \textit{goodput} -- i.e., the number of kilobytes
available for potential applications per second. 
The block size introduces a natural trade-off between the latency
and goodput. We observe that if we increase the block size from 1MB (which is the default for Bitcoin) to 2MB, the block stale rates remain similar and the goodput is roughly doubled. However, if we increase the block size to 4MB then the stale rates are so high for small $\Delta_1$ that throughput is not meaningfully higher than for 2MB blocks. As depicted, \name offers a goodput between 200 and 600 KB/s,
with a round latency between 5 and 6.5 seconds, respectively. In Bitcoin, where the size of a 2-in-2-out transaction is around 450 bytes, this would roughly correspond to between 450 and 1300 transactions per second.
In particular, we found
$\Delta_1=1.5s$ and $\Delta_2=4.0s$ as a promising configuration for 2MB blocks.

Further, we investigated the computation costs incurred by \name nodes.
We used a
single core of Intel i7 (3.5 GHz) CPU to compute vote and block processing
times.
On average, it takes only 53.25 $\mu$s to create a signed vote, 17.22 ms to
create a new block with 100 votes, and 10.08 ms to validate a received block
including 100 votes.

\begin{table}[t!]
\begin{center}
\setlength{\tabcolsep}{4pt}
    \caption{Performance results for different parameters.}
    \label{tab:perf}
    \footnotesize
\begin{tabular}{cc|rrc|rrc|rrc}
    &&\multicolumn{3}{c|}{1MB Blocks}&\multicolumn{3}{c|}{2MB Blocks}&\multicolumn{3}{c}{4MB Blocks}\\
    $\Delta_1$ & $\Delta_2$ & $\psi_b$ & $\psi_v$ & Gput. & $\psi_b$ & $\psi_v$ & Gput. & $\psi_b$ & $\psi_v$ & Gput.\\
    $s$ & $s$ & \%   & \%   & KB/s & \%   & \%   & KB/s & \%   & \%  & KB/s  \\\hline
    1.0 & 4.0 & 9.5 & 19.6 & 179 & 14.3 & 33.5 & 343 & 49.0 & 41.1 & 204 \\
    3.0 & 3.0 & 0   & 5.8  & 167 & 2.5  & 22.3 & 325 & 16.5 & 26.7 & 550 \\
    2.0 & 3.0 & 0   & 6.1  & 200 & 0    & 23.0 & 400 & 37.1 & 34.4 & 336 \\
    2.5 & 3.0 & 0   & 4.1  & 182 & 0    & 20.6 & 364 & 31.2 & 32.1 & 371 \\
    2.0 & 4.0 & 0   & 2.0  & 167 & 0    & 5.4  & 333 & 31.6 & 38.6 & 412 \\
    1.5 & 4.0 & 1.4 & 1.9  & 179 & 0    & 5.3  & 364 & 48.4 & 39.0 & 267 \\
    3.5 & 3.0 & 0   & 1.4  & 154 & 0    & 2.9  & 308 & 2.0  & 7.1  & 603 \\
    4.0 & 3.0 & 0   & 1.5  & 143 & 0    & 2.2  & 286 & 0    & 5.6  & 571 
\end{tabular}
\end{center}
\end{table}

\section{Comparison to Algorand}
\label{sec:algorand}
In this section we provide a detailed comparison to Algorand \cite{gilad2017algorand}, which is a closely-related
PoS protocol. Algorand shares several similarities with \name: 1) the protocol operates as a sequence of rounds, 2) leaders and committee members are selected in each round based on a random beacon that changes between rounds, 3) the likelihood that a node is selected as a leader or committee member is proportional to its number of stake units, 4) committee members vote on blocks proposed by the leaders, and 5) blocks are committed if they receive a sufficient number of votes. 

Despite these similarities, there are two main differences between Algorand and \name. The first is that instead of the cryptographic sampling procedure of \autoref{alg:election}, Algorand elects leaders and committee members in every round by running a VRF~\cite{micali1999verifiable} over the round's random beacon. As a result, the number of leaders and committee members in each round is random instead of fixed, and the added variance makes block commitment less secure. Second, instead of the voting and sequential hypothesis testing procedures of \autoref{alg:voting}~and~\ref{alg:commit}, Algorand uses a bespoke Byzantine agreement protocol (BA$\star$) to commit blocks. However, BA$\star$ includes many steps during which no transactions are added to the blockchain, and the security analysis for BA$\star$ is more restricted than for \name. There are other differences -- e.g., rewards and fairness are not discussed in \cite{gilad2017algorand} -- but due to space restrictions we only elaborate on the two main differences. A summary of the technical details of Algorand can be found in App.~\ref{sec:algorand technical}.

We first investigate what the impact would be if \autoref{alg:election} of \name were replaced by Algorand's committee selection procedure, so that $l$ and $q$ are no longer fixed. 
The main consequence is that the additional variance due to the committee sizes makes it harder to make block commitment decisions using hypothesis testing. 
This difference in terms of variance is made explicit in \autoref{tab:means_and_variances} of App.~\ref{app:vrf}. In particular, the quantity of supporting stake per round $X$ was previously hypergeometric (drawing $q$ samples without replacement from a population of size $n$ of which $u$ support the user's branch), whereas in the setting with VRFs it is binomial with a considerably larger population size (drawing $n$ samples such that each is in the committee with probability $q/n$ and supportive of the user's branch with probability $u/n$). The variance is typically around 3-4 times higher in the new setting if the adversary controls one third of the stake (i.e., $u/n = \frac{2}{3}$). 

As we can see from \autoref{fig:timetocommit}, the higher variance translates into block commitment times that are also roughly 3-4 times higher. For this experiment, we calculated numerically for different values of the average supporting stake per round $\bar{s}$ how many rounds it would take to commit in a setting where $p^* = 10^{-64}$, $n=1500$, $u=1000$, $q=150$, and $\gamma=0.99$ (as discussed in \autoref{sec:analysis:safety}). If $X$ is binomially distributed, then $T = \sum_{i=n+1}^m X_i$ is also binomially distributed, so we do not need to use the Cram\'er-Chernoff method to bound the probability $\P(T\geq t)$ for the VRF setting because we can calculate it directly. However, the looseness of the Cram\'er-Chernoff bound is clearly offset by the VRF method's higher variance. It can also be seen that if $\bar{s}$ is high enough, then \name can commit the blocks very quickly. For example, if $\bar{s}$ is at least equal to $98\%$, then \name commits within 3 rounds despite the very high security requirement (i.e., a safety error probability per round of $10^{-64}$). As long as $\bar{s}$ is above $86\%$, \name commits within 10 rounds -- with 5.5 seconds per round (see \autoref{sec:implementation}), this takes 55 seconds and therefore less time than the (on average) hour for 6 confirmations in Bitcoin.

\begin{figure}[t!]
\centering
  \includegraphics[width=.5\linewidth]{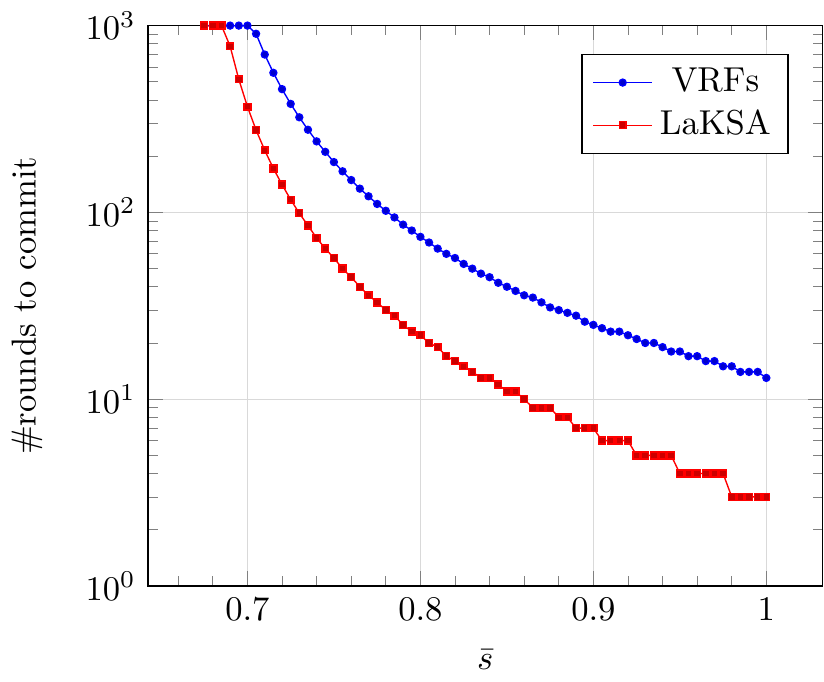}
\caption{Number of rounds needed before a block can be committed given an average supporting stake fraction per round $\bar{s}$ for fixed committee sizes (\name) and VRFs. It takes around 3-4 times longer to commit using VRFs.}
  \label{fig:timetocommit}
\end{figure}

Although the above demonstrates that VRF-based committee selection increases the number of rounds needed to commit a block using hypothesis testing, the value proposition of Algorand's BA$\star$ algorithm is that it commits blocks after a single round. However, there are several important caveats. The first is that even in the best-case scenario, the two phases of BA$\star$ take 2 steps each, so a round in Algorand takes at least $\lambda_{\textsc{priority}} + 4\cdot \lambda_{\textsc{step}} = $ 85 seconds (see also App.~\ref{sec:algorand technical}), which is the same as 15 rounds of \name. No transactions are added during the BA$\star$ steps, whereas in \name transaction-containing blocks are confirmed by other transaction-containing blocks. The second is that the security analysis relies on larger committee sizes (10,000 stake units in the final round) and a weaker adversary (who only controls $20\%$ of the stake) than \name. Finally, each user in \name is able to set her security threshold $p^*$ individually depending on her risk tolerance, whereas Algorand has a fixed set of security parameters that determine block commitment.

\section{Related Work}
\label{sec:related}
Besides Algorand, Bitcoin's NC~\cite{nakamoto2008bitcoin} has inspired a variety of alternative protocols that aim to build on its core strengths while mitigating its weaknesses, i.e., wastefulness, a slow and unclear commitment process, low transaction throughput, and a tendency to centralization. 
In this section, we compare \name to some of its most prominent alternatives, and highlight instances where they copy drawbacks from NC or introduce new ones. An overview is presented in \autoref{tab:related}.

\begin{table}[htp]
\begin{center}
\caption{Comparison of different blockchain protocols. HotStuff is a pure BFT design that can also be applied in a PoS setting.}
\label{tab:related}
\begin{tabular}{cccc}
 		 & Node      & Block & $\#$Leaders/Voters \\[-0.05cm]
Protocol & Selection & Commitment & Per Round \\ \toprule
NC \cite{nakamoto2008bitcoin} & PoW & Chain-Based & $\Theta(1)$ \\
Peercoin \cite{king2012ppcoin} & PoW/PoS & Chain-Based & $\Theta(1)$ \\
Tendermint \cite{kwon2014tendermint} & PoS & BFT Voting & $\Theta(n)$ \\
Casper FFG \cite{buterin2017casper} & PoS & BFT Voting & $\Theta(n)$ \\
HotStuff \cite{yin2018hotstuff} & --- & BFT Voting & $\Theta(n)$ \\
Algorand \cite{gilad2017algorand} & PoS & BFT Voting & $\Theta(q)$ \\
\name & PoS & Chain-Based & $\Theta(q)$ \\[-0.05cm] \bottomrule
\end{tabular}
\end{center}
\end{table}

The first PoS protocols that tried to extend or modify NC by
taking stake into consideration resulted in hybrid PoS-PoW
systems~\cite{king2012ppcoin,bentov2014proof}. The first of these, PPCoin (later renamed to Peercoin) \cite{king2012ppcoin}, extends Bitcoin by also granting the ability to propose blocks to nodes who hold on to tokens instead of spending them (i.e., increasing their \emph{coin age}). Bentov et al.\ instead propose a protocol~\cite{bentov2016cryptocurrencies} in which
new blocks generate randomness that is used to coordinate the extension of the blockchain
in the next several rounds. Other PoS
protocols~\cite{fan2017scalable,wang2019proof} try to emulate
the mining process of NC by using unique digital
signatures.
All of these systems mitigate Bitcoin's energy waste; however, they share multiple other drawbacks with NC,
such as an unclear commitment process and a tendency to centralization via high
reward variance.

To address the low throughput and unclear commitment process in NC, several other approaches replace NC's longest-chain rule by more established BFT consensus algorithms. The first approach to do so is Tendermint~\cite{kwon2014tendermint}, in which leaders are elected proportionally to their stake using a round-robin procedure.  The other nodes run a protocol that is based on Practical Byzantine Fault Tolerance (PBFT) \cite{castro1999practical} to agree on whether to commit the proposed block. Like most BFT protocols, PBFT works in the presence of $f$ adversarial nodes as long as $n \geq 3f + 1$. After a block is proposed, all nodes vote across two phases (\emph{prepare} and \emph{commit}), and a block is committed if at least $2f+1 \approx \frac{2}{3}n$ nodes vote to approve the block in both phases.  To avoid
attacks on leaders, which are known in advance, Tendermint proposes a lightweight
network-level anonymity solution \cite{kwon2014tendermint}. 

Casper FFG~\cite{buterin2017casper} is another PBFT-based PoS protocol, and functions as a finality-providing overlay for PoW or PoS blockchains. One of its main observations is
that the two phases of PBFT can be encoded in a blockchain as vote messages for two subsequent blocks.
Casper FFG is designed for an open setting with a dynamic set of nodes, and includes
incentive-based protection against equivocation -- i.e., misbehaving nodes risk losing
their deposits. In both Tendermint and Casper FFG, all nodes send their vote message to all other nodes in each voting round -- as such, when $n$ grows large the communication complexity becomes a bottleneck.
 HotStuff~\cite{yin2018hotstuff} considerably reduces the message complexity of  protocols such as Tendermint and Casper FFG through the use of threshold signatures. However, in HotStuff all participating nodes still vote in every round of the protocol.

Another approach to reduce the message complexity of BFT voting in each round is to draw a committee from the total set of nodes instead of requiring that all nodes vote. However, the safety properties of PBFT-based approaches do not straightforwardly carry over to this setting. The most obvious generalization of PBFT to random committees is to require that a block receives more than $\frac{2}{3}q$ votes instead of $\approx \frac{2}{3}n$ votes across two rounds for it to be committed, but this can lead to high safety fault probabilities -- see also \autoref{sec:committee_bft}.
As discussed in \autoref{sec:algorand}, Algorand \cite{gilad2017algorand} introduces a bespoke BFT algorithm called BA$\star$, which addresses the above problem by 1) requiring a higher fraction than $\frac{2}{3}$ of supporting votes to commit ($0.74q$ in the final round), 2) requiring large committee sizes, and 3) assuming limited adversarial strength.

\subsection{Message Complexity per Round}
\label{sec:complexity}

NC has a message complexity of $\Theta(n)$ in each round, as leaders are elected
without requiring any messages, and a leader sends her block to the other $n-1$
nodes. PBFT-like protocols such as Tendermint and Casper FFG have a message
complexity of $\Theta(n^2)$, because all nodes except the leader vote in each
round, and each vote is sent to every other node. In HotStuff, voters only send
their votes to the leader, and the leader then sends the block including an
aggregated signature to the nodes, resulting in a message complexity of
$\Theta(n)$. Algorand has a message complexity of $\Theta(qn)$, as each
committee member sends a message to the other $n-1$ nodes. Since the leader is
not known -- which is inescapable in the first phase, as the winners of the
VRF-based election are unknowable by design -- this cannot be reduced using the
approach of HotStuff. In the implementation of \name as presented in
\autoref{alg:voting}, the communication complexity is also $\Theta(qn)$
since each vote is sent to $n-1$ nodes. However, since the leader is known in
each round, this can be reduced to $\Theta(n)$ using the approach of HotStuff.
One adverse effect is that this makes the creation of virtual blocks in the case
of offline or malicious leaders more complicated -- essentially, voters
must also send their votes to the leaders of later rounds. 

\subsection{Other Related Work}
\label{sec:misc related}
Brown-Cohen et al.~\cite{brown2018formal} analyze longest-chain PoS protocols
and show their fundamental limitations in preventing PoS-specific attacks. These results do not directly apply to \name, since chain selection in \name fully depends on votes (and not leader-driven chain length) and a GHOST-like rule which cannot be
expressed in the framework of \cite{brown2018formal}; both aspects are mentioned by the authors of \cite{brown2018formal} as limitations.
The Ouroboros family of protocols \cite{kiayias2017ouroboros,david2018ouroboros,badertscher2018ouroboros,kiayias2018ouroboros} also uses a PoS approach with committee voting. Ouroboros shares the limitations of longest-chain PoS protocols as reported by Brown-Cohen et al.~\cite{brown2018formal}.
Moreover, the protocol leaves rewards and incentives as future work. Cardano, a cryptocurrency built upon Ouroboros, encourages users to either join or create stake pools~\cite{cardano-staking} (see also \cite{brunjes2018reward}), thus encouraging centralization by design.
Under the Ethereum 2.0 roadmap, Ethereum's legacy PoW chain will be phased out in favor of a novel PoS chain that features Casper FFG and sharding. The \textit{beacon chain}, which will eventually coordinate the shard chains, was launched in December 2020. The beacon chain's block proposal mechanism \cite{buterin2020combining} shares some similarities with \name, e.g., it is also chain-based, with \textit{attestations} taking the place of votes in the fork-choice rule. In this approach, time is divided into epochs that consist of 32 rounds -- in each epoch, the set of active participants is pseudorandomly shuffled and divided among the epoch's slots to ensure that each node votes once per epoch.
Finally, several recent surveys \cite{cachin2017blockchain,DBLP:journals/corr/abs-1711-03936,natoli2019deconstructing} present an overview of recently proposed blockchain protocols.

\section{Conclusions}
\label{sec:conclusions}
In this work we have proposed \name, a novel PoS consensus protocol dedicated to
cryptocurrencies.  Surprisingly, through its simple construction \name provides a robust, scalable, and
secure consensus mechanism.  Our scheme extends the notion of probabilistic safety -- in particular, clients base block commitment decisions on the probability of the block being reverted given the total observed support for it.
These decisions are made more precise thanks to a lightweight committee voting scheme
that allows large numbers of nodes to participate and 
express their beliefs about the blockchain. 

In this work we have presented the core concept behind \name and its properties.  In
the future, we plan to extend and analyze the system in a more dynamic setting
and with adaptive adversaries. In particular, we believe that the ideas that are present
in recent protocols~\cite{gilad2017algorand,daian2017snow,david2018ouroboros}
can be successfully applied in \name, enhancing the protocol further.  Another
interesting research problem is to find an efficient election protocol combining
the advantages of the proposed schemes (i.e., ``secret'' but deterministic
election).  We also plan to study further economic aspects of the reward scheme
and their influence on the security of the system.

\bibliographystyle{abbrv}
\bibliography{ref}

\appendices
\label{sec:appendix}
\section{Cryptographic Sampling Proof}
\label{app:prf-sampling}
 We show that no probabilistic polynomial-time (PPT) adversary can distinguish
 outputs of our cryptographic sampling (see \autoref{sec:details:election}) from
 truly random sampling.

\begin{definition}[Pseudorandom sampling]
$PRF$ is a pseudorandom sampling and mapping $\{0,1\}^n \times \{0,1\}^s \to
    \{0,1\}^n$ if it is collision-resistant and satisfies the following requirements:
\begin{inparaenum}
  \item For any $k \in \{0,1\}^s$, $PRF$ is a bijection from $\{0,1\}^n$ to $\{0,1\}^n$.
  \item For any $k \in \{0,1\}^s$, there is an efficient algorithm to evaluate $PRF_k(x)$.
  \item For any PPT distinguisher $\mathcal{D}$:
      \begin{equation}\label{eq:permutation}
        |Pr(\mathcal{D}^{PRF_k(\cdot)}(1^n) = 1)-Pr(\mathcal{D}^{f_n}(1^n) = 1)|< negl(s),
      \end{equation}
      where $k \gets \{0,1\}^n$ is chosen uniformly at random and $f_n$ is
chosen uniformly at random from the set of permutations on $n$-bit strings.
\end{inparaenum}
\end{definition}

\begin{corollary}
If the output of $PRF_r(\cdot)$ is indistinguishable from the uniform distribution,
    then the result of $PRF_{r}(\cdot)\% N$ for any $N$ is
    indistinguishable from the uniform distribution.
\end{corollary}

In particular, a pseudorandom sampling family is a collection of pseudorandom functions, where a specific sampling may be chosen using a key as a \emph{salt}.
\begin{lemma}
We say that $PRF$ is an unpredictable pseudorandom sampling, if no
PPT adversary can distinguish the unit-pair from a uniform random distribution.
\end{lemma}
\begin{proof}
For all PPT distinguishers $\mathcal{D}$:
      \begin{equation*}
        |Pr(\mathcal{D}^{PRF_k(\cdot)}(1^n) = 1)-Pr(\mathcal{D}^{f_n}(1^n) = 1)|< negl(s),
      \end{equation*}

 In other words, a PRF is any pseudorandom hash function that can be used to map data of arbitrary size to data of fixed size, and the above equation implies that if the PRF is secure, then its output is indistinguishable from random output.
We define the distributions $\mathcal{X}$  and $\mathcal{Y}$ as follows.
  $\mathcal{X}$ is the distribution
    $o_i^1, o_i^2,\cdots, o_i^S$,
  where $o_i^j \gets PRF_{r_s}\big( i \|role \big)$, and  $\mathcal{Y}$ is the uniform distribution.
Then the two distributions $\mathcal{X}$ and $\mathcal{Y}$ are computationally
indistinguishable.
Below, we prove it by contradiction.

We assume that $\mathcal{D}$ is a PPT adversary who can distinguish
$\mathcal{X}$ from $\mathcal{Y}$ with non-negligible advantage. For $1\leq i
\leq S+1$, $S$ is the number of stake units, we introduce intermediate distributions
$\mathcal{X}_i$ that are given by
  $\overline{h}^1 ,\cdots , \overline{h}^{i-1} , {h}^i , \cdots , {h}^S$, 
where ${h}^i$ is as above and $\overline{h}^i$ is uniformly chosen from $\mathbb{Z}_q$. Hence we obtain $\mathcal{X}_1=\mathcal{X}$ and $\mathcal{X}_{S+1}=\mathcal{Y}$.
By assumption, $\mathcal{D}$ can distinguish $\mathcal{X}_1$ from $\mathcal{X}_{S+1}=\mathcal{Y}$ with noticeable (or overwhelming) advantage $\epsilon$ and so, by a standard hybrid argument, there is some $i$ such that $\mathcal{D}$ can distinguish $\mathcal{X}_i$ from $\mathcal{X}_{i+1}$ with some noticeable advantage at least $\epsilon/S$.
It is then easy to see that $\mathcal{D}$ gives a  distinguisher.
That is, \autoref{eq:permutation} does not hold.
By assumption, no such distinguisher exists, and hence the lemma is proved.
\end{proof}

\begin{corollary}
If the pseudorandom sampling distribution
 $PRF_{r}(1\|role)\|\cdots\| PRF_{r}(Q \|role)$
 is indistinguishable from the uniform distribution, then the sampling  distribution
 $PRF_{r}(1\|role)\ \%\ Len(tmp_1)\| \cdots  \| PRF_{r}(Q \|role) \ \%\ Len(tmp_{Q})$
is indistinguishable from uniform, where $tmp_{i}$ for $i\in[1, Q]$.
\end{corollary}
We remark that $tmp_{i}$ is dynamically changing for $i\in[1, Q]$, but it does not effect the indistinguishability result because $PRF_{r}(i\|role)$ is pseudorandom, so $PRF_{r}(i\|role)\% N$ is also pseudorandom.

\section{Probability Computation Algorithms}

\autoref{alg:helpers} presents baseline implementations of the helper functions needed to apply \autoref{alg:probcomphyper}. They respectively compute the logarithm of the binomial coefficient (\textit{LogBinCoef}), the probability mass function of the hypergeometric distribution (\textit{HypergeometricPMF}), the log of the Gaussian/ordinary hypergeometric function (\textit{Hyp2F1}), a suitable base interval for golden-section search with RateFuncHelper as defined in \autoref{alg:probcomphyper} (\textit{RateFuncSearchRange}), and the new interval after a step of golden-section search (or rather a variation that uses two simpler mid points, for presentational simplicity). \textit{HypergeometricPMF} uses the logarithm of the binomial coefficient to avoid floating-point errors during computation. The function LogHyp2F1 computes the logarithm of $_2F_1(a,b,c,z)$, which is defined as
$
_2F_1(a,b,c,z) = \sum_{n=0}^{\infty} \frac{(a)_n (b)_n}{(c)_n} \frac{z^n}{n!}
$
where 
$
(x)_n =  \prod_{i=0}^{n-1} (x+i) $
for $n \in \mathbb{N}$ and $z \in \mathbb{R}$. It similarly uses logarithms when possible in order to avoid numerical errors.
\begin{algorithm}[t!]
\caption{Helper functions.}
\label{alg:helpers}
\footnotesize
\SetKwProg{func}{function}{}{}
\func{LogBinCoef(n,k)}{
    $x \leftarrow 0$;
	$m \gets \min(k, n-k)$;
	
	\For{$i \in \{0,\ldots,m-1\}$} {
		$x \leftarrow x + \log(n-i) - \log(m-i)$;
	}
	\Return x;
}

\func{HypergeometricPMF($S$,$M$,$Q$,$i$)}{
	$x \leftarrow \textit{LogBinCoef}(M,i) - \textit{LogBinCoef}(S-M,Q-i)$ \\ \hskip0.6cm $-$ $\textit{LogBinCoef}(S,Q)$;

	\Return $e^x$;
}

\func{LogHyp2F1($a$,$b$,$c$,$d$)}{
	$x \leftarrow 1$;
    $y \leftarrow \infty$;
    $n \leftarrow 1$;
    $r \leftarrow 0$;

    // $\epsilon$ is an accuracy threshold set by the user

    \While{$e^{y + r} < \epsilon$}{
    	$y \leftarrow 0$;
	
		\For{$i \in \{1,\ldots,n\}$}{
			$y \leftarrow y + \log(-a-i+1)-\log(-c-i+1)$
			$+\log(-b-i+1)-\log(i)+\log(z)$;
		}
		\If{$y > -\infty$}{
			$k \leftarrow \lfloor y \rfloor$;
			$r \leftarrow r + k;$
			$x \leftarrow x e^{-k} + e^{y - k}$;
		}
		$n \leftarrow n + 1$;
    }

    \Return $\log(x) + r$;
}

\func{RateFuncSearchRange($n$,$u$,$q$,$t$)}{
	$\vec{\lambda} \leftarrow (0,1,2,3)$;
	
	\For{$i \in \{1,\ldots,4\}$}{
		$y_i \gets \textit{RateFuncHelper}(n,u,q,\lambda_i,t)$;
	}
	
	\While{$y_3 > y_2$ and $y_2 > y_1$ and $y_1 > y_0$} {
	
		\For{$i \in \{1,\ldots,4\}$}{
			$\lambda_i \leftarrow 2 * \lambda_i$;
			$y_i \gets \textit{RateFuncHelper}(n,u,q,\lambda_i,t)$;
		}
	}
	\Return $y_3$;
}

\func{ReduceInterval($\vec{\lambda}$,$\vec{y}$,$n$,$u$,$q$,$t$)}{
	
	\If{$y_1 > y_2$}{
		$\lambda_3 \leftarrow \lambda_2$;
		$\lambda_2 \leftarrow 2/3 * \lambda_3 + 1/3 * \lambda_0$;
		$\lambda_1 \leftarrow 1/3 * \lambda_3 + 2/3 * \lambda_0$;
	}
	\Else{
		$\lambda_0 \leftarrow \lambda_1$;
		$\lambda_1 \leftarrow 2/3 * \lambda_0 + 1/3 * \lambda_3$;
		$\lambda_2 \leftarrow 1/3 * \lambda_0 + 2/3 * \lambda_3$;
	}
	
	\Return $\vec{\lambda}$;
}

\end{algorithm}

\section{Safety Analysis Lemmas}
The following lemma gives an expression for the probability mass function of sums of hypergeometric random variables.
\begin{lemma}
Let $X_1,\ldots,X_k$ be independent and identically distributed, such that each $X_i$ takes values on $\{0,1,\ldots,q\}$. Let $T_k = \sum_{i=1}^k X_i$. Then
$$
\P(T_k = t) = \sum_{x_2=0}^q \sum_{x_3=0}^q \ldots \sum_{x_{k}=0}^q \prod_{i=2}^{k} \P(X_i = x_i) \P\left(X_{1} = t - \sum_{i=2}^{k} x_i\right).
$$
\label{lm:convolutions}
\end{lemma}

\begin{proof} We prove this using mathematical induction. For the base case, we have that
$$
\P(T_2 = t) = \sum_{x_2=0}^q \P(X_2 = x_2) \P(X_1 = t - x_2)
$$
from the definition of the convolution of two random variables. 
\footnote{\url{https://en.wikipedia.org/wiki/Convolution_of_probability_distributions}}
For the induction step we have that
\begin{equation*}
\begin{split}
\P(T_k = t) = & \sum_{x_k=0}^q \P(X_k = x_k) \P(T_{k-1} = t - x_k) \\
= & \sum_{x_k=0}^q \P(X_k = x_k) \sum_{x_2=0}^q \sum_{x_3=0}^q \ldots \sum_{x_{k-1}=0}^q \\
& \prod_{i=2}^{k-1} \P(X_i = x_i)   \P\left(X_{1} = t - x_k - \sum_{i=2}^{k-1} x_i\right) \\
= & \sum_{x_2=0}^q \sum_{x_3=0}^q \ldots \sum_{x_{k}=0}^q \prod_{i=2}^{k} \P(X_i = x_i) \P\left(X_{1} = t - \sum_{i=2}^{k} x_i\right),
\end{split}
\end{equation*}
 where the first equality follows the definition of the convolution of two random variables, the second equality is the induction step, and the third equality holds because we are allowed to interchange the summations (as they are all over a finite set of elements). This proves the lemma.
 \end{proof}
The following lemma about the convexity of $c_X$ is needed at two different points in the text (namely for the validity of golden-section search, and for Lemma~\ref{lm:meansmalllft}). It is a well-known result -- 
see, e.g., \cite{boucheron2013concentration} or Cosma Shalizi's lecture notes on stochastic processes (week 31), which can be found \href{https://www.stat.cmu.edu/~cshalizi/754/2006/}{online} -- but it is included here for completeness.

\begin{lemma}
Let $X$ be a random variable. Its cumulant-generating function $c_X(\lambda) = \log(\E(e^{\lambda X}))$ is convex.
\label{lm:convexcgf}
\end{lemma}
\begin{proof}
The lemma follows from a simple application of H\"older's inequality.\footnote{\href{https://en.wikipedia.org/wiki/H\%C3\%B6lder\%27s\_inequality}{https://en.wikipedia.org/wiki/H\"older's\_inequality}} 
\begin{equation*}
\begin{split}
c_X(a \lambda + b \mu)  & = \log\E\left( e^{(a \lambda + b \mu) X } \right) \\
& = \log\E\left[ \left(e^{\lambda X }\right)^a \left(e^{\mu X }\right)^b\right] 
\\ 
     & \leq \log\left(\E\left[ \left(e^{\lambda X }\right)^a \right] \cdot \E\left[ \left(e^{\mu X }\right)^b\right]\right) \\
 & = a c_X(\lambda) + b c_X(\mu)
\end{split}
\end{equation*}
\end{proof}

The following lemma asserts that when the per-round supporting stake does not exceed the expected value of our distribution (which in our case is given by $q u / n$), then we cannot establish a meaningful bound using the Cram\'er-Chernoff method. However, if it \emph{does}, then the rate function is positive, which means that the bound (and therefore the $p$-value) goes to zero as the number of rounds increases.

\begin{lemma}
Let $X$ be a random variable. Then its large-deviations rate function $r_X(t)$, defined as
$
r_{X}(t) = \sup_{\lambda \geq 0} (\phi_t(\lambda)) = \sup_{\lambda \geq 0} (\lambda t - c_{X}(\lambda)),
$
has the following properties:
\begin{inparaenum}[1)]
\item if $t \leq \E(X)$, then $r_{X}(t) = 0$, and \label{it:propone}
\item if $t > \E(X)$, then $r_{X}(t) > 0$. \label{it:proptwo}
\end{inparaenum}
\label{lm:meansmalllft}
\end{lemma}

\begin{proof}
By definition, $c_{X}(0) = \log (1) = 0$. Hence, \mbox{$\phi_t(0) = 0$} for all $t$, so by the definition of the supremum it must hold that $r_{X}(t) = \sup_{\lambda \geq 0} (\phi_t(\lambda)) \geq 0$. In the following, we will investigate $\phi_t'(0)$, i.e., the derivative of $\phi_{t}(\lambda)$ at $\lambda = 0$. If $\phi'_{t} (0) \leq 0$, then by the convexity of $c_{X}$ we know that $\phi'_{t}(\lambda)\leq 0$ for all $\lambda \geq 0$, which proves property~\ref{it:propone}. However, if $\phi'_{t}(0) > 0$ then there must exist a point $\lambda^*>0$ such that $\phi_t(\lambda^*) > \phi_t(0) = 0$, and the supremum must at least be equal to this value.
Hence, we can demonstrate both properties of the lemma using $\phi'_{t} (0)$.

We find that
$$
\phi'_{t} (\lambda) = \frac{d}{d\lambda} \left(\lambda t - \log(\E(e^{\lambda X})) \right) = t - \frac{\frac{d}{d\lambda} \E(e^{\lambda X})}{\E(e^{\lambda X})}
$$
because of the chain rule. Hence,
$$
\phi'_{t} (0) = \left. t - \frac{\frac{d}{d\lambda} \E(e^{\lambda X})}{\E(e^{\lambda X})} \right\vert_{\lambda = 0} = t - \E(X)
$$
because $\E(e^{0}) = 1$, and $\left. \frac{d}{d \lambda} \E(e^{\lambda X}) \right\vert_{\lambda = 0} = \E(X)$ because of the fundamental property of the moment-generating function\footnote{\url{https://en.wikipedia.org/wiki/Moment-generating\_function}} that $\left. \frac{d^n}{d^n \lambda} \E(e^{\lambda X}) \right\vert_{\lambda = 0} = \E(X^n)$.

Obviously, if $\phi'_{t} (\lambda) = t - \E(X)$ then $\phi'_t(0) \leq 0$ if $t \leq \E(X)$ and $\phi'_t(0) > 0$ otherwise, which proves the lemma.
\end{proof}

\section{Liveness Analysis Lemmas}
\label{app:liveness}

\begin{lemma}
    \label{lemma:strongchain}
    The main chain after $t$ rounds includes at least $t(1-\alpha)q$
    stake units while $t$ increases.
\end{lemma}
\begin{proof}
    While $t$ increases, the expected number of ``honest'' stake units elected
    per round is $(1-\alpha)q$.  When the network is synchronous, honest nodes
    vote for the same blocks and their votes are delivered on time.  Thus, with
    the failing (i.e., not voting) adversary, the strongest chain will obtain
    $t(1-\alpha)q$ of the supporting stake.  An adversary can cause honest nodes to
    change their current main chain, e.g., by showing a stronger fork, but
    according to our chain selection rule (\autoref{sec:details:forks}) such a
    chain would need to contain at least the same amount of supporting stake as
    their main chain.  Thus, the final main chain would include at least
    $t(1-\alpha)q$ stake units.
\end{proof}

\begin{lemma}
    \label{lemma:nchain}
    \label{lemma:ngrowth}
    With increasing $t$, the main chain after $t$ rounds includes $m$ blocks
    that the honest participants voted for, where
        $$t\geq m\geq \lceil t(1-2\alpha)\rceil.$$
\end{lemma}
\begin{proof}
    Following Lemma~\ref{lemma:strongchain}, honest nodes after $t$ would agree
    on the main chain $C$ that has at least $t(1-\alpha)q$ stake units.  Let us
    assume that after $t$ rounds there exists an adversarial chain $C'$ with
    total stake of $t\alpha q+s$, where $t \alpha q$ is adversarial stake and $s$ is
    the stake shared between $C$ and $C'$ (i.e., the chains have common
    blocks/prefix).  $C'$ can overwrite $C$ only if: $t\alpha q+s\geq t(1-\alpha)q,$
    giving the bound for the shared stake:  
    $s\geq t(1-\alpha)q-t\alpha q.$
    As every block contains maximally $q$ stake units, $s$ can be
    contributed in minimum $\lceil\frac{t(1-\alpha)q-t\alpha q}{q}\rceil$ blocks
    shared between $C$ and $C'$.  Thus, 
    $n\geq \lceil t(1-2\alpha)\rceil.$

    Note, that $n$ is upper bound by $t\geq n$, as the main chain cannot contain
    more than $t$ blocks in $t$ rounds.
\end{proof}

\section{A Naive BFT Algorithm with Random Committees}
\label{sec:committee_bft}

In this section we show that BFT cannot be straightforwardly generalized to a setting with randomly sampled committees. We illustrate this using the worst-case adversarial setting of \autoref{sec:analysis:safety}, in which the honest nodes are split evenly across a network partition: i.e., a malicious leader proposes two conflicting blocks, and $n+1$ nodes are aware of one block and $n+1$ of the other. In the normal setting of BFT protocols (i.e., without committees), the safety property holds: even if the remaining $n-1$ adversarial nodes equivocate by voting for both blocks, then both blocks receive only $2f$ votes, so neither is committed and no safety fault occurs. However, when only a randomly drawn committee is allowed to vote, the probability that the committee contains enough malicious voters to commit conflicting blocks can be very high. 

We consider the following setting: again, an adversarial user proposes two conflicting blocks $A$ and $B$, such that $n_A$ honest users vote for block $A$ and $n_B$ honest users vote for block $B$. The remaining $f-1$ malicious users vote for both blocks. Let $n_A = n_B =  f+1$, so that \mbox{$n = f + n_A + n_B = 3f + 1$}, which means that the number of honest nodes is sufficient for Byzantine fault tolerance. A committee of size $q$ is drawn that includes $X_{f}$ malicious users, $X_{A}$ honest users who vote for block $A$, and $X_B$ users vote for block $B$, such that $X_f$, $X_A$, and $X_B$ are random variables. Due to the analogy with drawing balls from an urn, the joint probability distribution of $X_f$, $X_A$, and $X_B$ can be seen to be the multivariate hypergeometric distribution, i.e.,
$$
\P(X_f = x_f, X_A = x_A, X_B = x_B) = \frac{\comb{X_f}{x_f}\comb{X_A}{x_A}\comb{X_B}{x_B}}{\comb{n}{q}}.
$$ 
A straightforward generalization of PBFT would commit a block if it receives $T = \lceil \frac{2fq}{n} \rceil + 1$ instead of $2f + 1$ votes over two rounds. If the network is in a synchronous period, so that all votes that are cast are also received by all other nodes within the round, then a safety fault occurs if \mbox{$X_A + X_M \geq T$} and \mbox{$X_B + X_M \geq T$}. This probability can be calculated directly in \autoref{alg:naivebft}. The results of this algorithm for different choices of $f$ and $q$ are given in \autoref{tab:naivebft}. We observe that the probability of a safety fault grows as the number of nodes increases, and that this probability is greater than $10\%$ per round when $f \geq 1000$ (and hence $n \geq 3001$) even if $90\%$ of the nodes are sampled for the committee.

\begin{algorithm}[!t]
\caption{Naive Committee BFT}
\label{alg:naivebft}
\footnotesize
\SetKwProg{func}{function}{}{}
\func{CalcNaiveBFTFaultProb(f,q)}{
	$n_A \gets f+1$;
	$n_B \gets f+1$;
	$n \gets f + n_A + n_B$;
	$T \gets \lceil 2 f q /n\rceil$
	$p \gets 0$
	
	\For{$x_f \in \{0,\ldots,q\}$} {
		\For{$x_A \in \{0,\ldots,q-f\}$} {
		
			$x_B \gets q - m - x_A$
		
			\If{$x_f + x_A \geq T$ \textbf{and} $x_f + n_B \geq T$} {
				\If{$x_f \leq f$ \textbf{and} $x_A \leq n_A$ \textbf{and} $x_B \leq n_B$} {
					$y \gets \textit{LogBinCoef}(f-1,x_f) + \textit{LogBinCoef}(n_A,x_A) + \textit{LogBinCoef}(n_B,x_B) - \textit{LogBinCoef}(n,q)$
					 
					 $p \gets p + e^{y}$
				}
			}
		}
	}
	\Return p;
}

\end{algorithm}

\begin{table}[t!]
    \caption{Probability of a safety fault in the naive BFT procedure of \autoref{sec:committee_bft}.}
\begin{center}
\begin{tabular}{c|ccccc}
& \multicolumn{4}{c}{\raisebox{0.05cm}{$Q/N$}} \\
$M$ & $0.1$ & $0.3$ & $0.5$ & $0.7$ & $0.9$ \\ \toprule
30 & 0.0544 & 0.064 & 0.0548 & 0.0358 & 0.0045 \\ 
100 & 0.0946 & 0.103 & 0.0961 & 0.0796 & 0.0383 \\ 
300 & 0.1219 & 0.1278 & 0.1232 & 0.1118 & 0.0783 \\ 
1000 & 0.141 & 0.1447 & 0.142 & 0.1351 & 0.1134 \\ 
3000 & 0.1516 & 0.1537 & 0.1521 & 0.148 & 0.1345 \\ 

\end{tabular}
\end{center}
\label{tab:naivebft}
\end{table}

\section{Technical Details of Algorand}
\label{sec:algorand technical}
In each round of Algorand, each node uses a VRF with the user's private key to calculate a hash of the random beacon concatenated with a ``role'' index. Each round in Algorand consists of two phases, and there are three roles: block proposer, initial (``step'') committee member, and final committee member. The value of the calculated hash determines how many of the node's stake units qualify for each role. For any node, let $N$ be the number of stake units that it controls, and $\nrole$ be the number of stake units that are assigned to a given role. Alg.~1 of \cite{gilad2017algorand} then ensures that $\nrole$ is binomially distributed with
$\displaystyle 
\P(\nrole = n) = \comb{N}{n} {\prole}^{n} (1-\prole)^{N - n},
$ 
where $\prole$ is the probability for each single stake unit that it is assigned to \textsc{role}. The probabilities are chosen such that the expected total numbers of block proposers, initial committee members, and final committee members per round equal $\tau_{\textsc{proposer}}$, $\tau_{\textsc{step}}$, and $\tau_{\textsc{final}}$, respectively. 

Each leader proposes a block such that there is a priority relation between the proposed blocks that is obvious to all nodes. The committee members wait a fixed amount of time per round ($\lambda_{\textsc{priority}}$) to receive blocks. Next, they initiate Phase 1 in which they vote for the highest-priority block that they are aware of (Alg.~7 of \cite{gilad2017algorand}). If a block receives more than $T_{\textsc{step}} \cdot \tau_{\textsc{step}}$ votes across two rounds of voting (``steps''), it advances to Phase 2. Phase 2 (Alg.~8 of \cite{gilad2017algorand}) consists of a sequence of steps in which the initial committee members either vote for the block from Phase 1 or an empty block. If either option receives more than $T_{\textsc{step}} \cdot \tau_{\textsc{step}}$ votes during a step, it is subjected to a final vote among the final committee members.
If the block receives more than $T_{\textsc{final}} \cdot \tau_{\textsc{final}}$ votes in the final vote, then it is committed by the nodes. Each step in either phase of BA$\star$ takes $\lambda_{\textsc{step}}$ time units. The Algorand white paper provides a set of benchmark parameters for BA$\star$ (see Fig.~4 of~\cite{gilad2017algorand}) to ensure that the probability of a safety fault is sufficiently small (less than $10^{-7}$ over 1000 rounds) if the adversary controls at most $20\%$ of the stake: \mbox{$\lambda_{\textsc{priority}} = 5$ seconds}, \mbox{$\lambda_{\textsc{step}} = 20$ seconds}, $\tau_{\textsc{proposer}} = 26$, $\tau_{\textsc{step}} = 2000$, $\tau_{\textsc{final}} = 10000$, $T_{\textsc{step}} = 0.685$, and $T_{\textsc{final}} = 0.74$.

\section{Stake Management}
\label{app:disc:stake}
Current PoS systems suffer from the undesired ``rich get richer''
dynamics~\cite{fanti2018compounding}.  In short, nodes with large stake holdings
obtain more stake from rewards (thus, more voting power) and effectively
quickly dominate the system.  Although that issue seems to be inherent to PoS
systems, we propose to divide stake into \textit{voting stake} and
\textit{transaction stake}.  The voting stake would be used only for running the
consensus protocol, while transaction stake would be for conducting standard
cryptocurrency transactions.  All rewards are paid in transaction stake, and not
in voting stake, which however, can be exchanged between nodes freely.  With this
simple modification, the system does not make rich nodes richer automatically,
instead a market for the voting stake is introduced.
The voting stake can be transacted (to allow newcomer nodes), however, we restrict
that the voting right is granted only after the transaction is confirmed by a
significant number of blocks such that its reversal probability is negligible.

A related issue is the voting stake that is not used for a long time period.  As
the system has to be performant and sustainable in the long-term, the traditional
model of failures may be insufficient as honest nodes can simply leave the
system due to other reasons than failures.  The related work
suggests~\cite{buterin2017casper} a solution where the stake that has been
inactive for a long time period may ``expire'' (as it may indicate that its
owner lost her keys or interests in supporting the platform).  We believe that a
similar can be applied successfully with \name.

\section{Election via VRF}
\label{app:vrf}
A verifiable random function (VRF)~\cite{micali1999verifiable} allows for the
generation of a pseudorandom output from a message and a secret key. The output can
 then be publicly verified by any party with the corresponding public key.
Algorand~\cite{gilad2017algorand} proposed a VRF-based method for committee and
leaders election dedicated to PoS systems.  
The core of the procedure is the cryptographic sortition
algorithm (see \S~5 of \cite{gilad2017algorand}).  Every node runs the
procedure locally to find out how many of its stake units were sampled within the
round. The likelihood of being elected is proportional to the stake possession.
This is similar to our cryptographic sampling construction
(\autoref{sec:details:election}), and Algorand's sortition can be applied in
\name almost directly, as we sketch below.
The main voting procedure is similar as previously, i.e., as in
\autoref{alg:vrf}.  The main difference is to use the sortition algorithm to
check whether the node is a leader and voter in the round.  Leaders and voters
add proofs to their blocks and votes, proving that they  were indeed elected
in the given round. These proofs have to be checked by other nodes by running
\textit{VerifyRole}.

\begin{algorithm}[t!]
\caption{The voting procedure with VRF election.}
\label{alg:vrf}
\footnotesize
\SetKwProg{func}{function}{}{}
\func{VotingRound(i)}{ 
    $r \leftarrow RoundBeacon(i)$; 
    $w \gets stake[pk]$\;
    $(h,\pi, s)\leftarrow Sort(sk, r, q, \textit{`vote'}, w, n)$; // Alg. 1
    in \cite{gilad2017algorand}\\ 
    \If(// check if I'm a voter){$s>0$}{
        $B_{-1} \leftarrow MainChain().lastBlk$; // get last block\\
        $\sigma\leftarrow Sign_{sk}(i\|H(B_{-1})\|h\|\pi\|s)$\;
        $v\leftarrow (i,H(B_{-1}), h, \pi, s, pk, \sigma)$; // support vote\\
        Broadcast($v$)\;
    }
    Wait($\Delta$); // meantime collect and verify support votes\\
    $(h,\pi, s)\leftarrow Sort(sk, r, L, \textit{`lead'}, w, n)$; // Alg. 1
    in \cite{gilad2017algorand}\\ 
    \If(// check if I'm a leader){$s>0$}{
        $B_{-1} \leftarrow MainChain().lastBlk$; // possibly different block\\
        $V\leftarrow\{v_a, v_b, v_c, ...\}$; // received $B_{-1}$'s support
     votes\\
        $r_i \leftarrow Random()$\;
        $\sigma\leftarrow Sign_{sk}(i \| r_i\| H(B_{-1}) \|F\| V \| \mathit{Txs}\|h\|\pi\|s)$\;
        $B \leftarrow (i, r_i, H(B_{-1}),F, V, \mathit{Txs}, h, \pi, s, pk, \sigma)$; // new block\\
         Broadcast($B$)\;
     }
     Wait($\Delta$); // wait for the next round
}
\func(// vrfy. leader/voter){VerifyRole(r, role, pk, h, $\pi$, s, size)}{
    $w \gets stake[pk]$\;
    $s'\gets \textit{VerifySort}(pk, h, \pi, r, size, role, w, n)$; // Alg. 2 in \cite{gilad2017algorand}\\
    \Return $s=s'$\;
}
\end{algorithm}

The main advantage of this solution is that voters and leaders are
unpredictable and become known only after publishing their messages.  This
eliminates some classes of adaptive attacks as an adversary cannot
target nodes at the beginning of the round.  
On the other hand, $q$ and $l$ are not constant in this setting and become random variables instead.
This effect may be undesirable as the amount of stake sampled in every round is
unknown, which influences commit decisions.\footnote{It may also decrease the
throughput as with random $l$ there may be no or multiple leaders per round.}
Our analysis can be generalized to this case, however -- the distribution of the supporting stake is no longer hypergeometric but binomial, and since sums of binomially distributed random variables are again binomial, the $p$-values can be computed directly without the need for bounding techniques such as the Cram\'er-Chernoff method. However, the higher variance means that it will take more rounds to commit, so the approach with fixed committee sizes provides faster finalization.
In particular, if $X$ is the number of stake units per round, then the variance
in the hypergeometric distribution of \autoref{sec:analysis} is given by 
$$
\text{Var}_{\textsc{HPG}} = \displaystyle q \frac{u}{n} \frac{n-u}{n} \frac{n-q}{n-1}.
$$ 
By contrast, if we use VRFs then every stake unit is chosen as part of the committee with probability $\frac{q}{n}$ and part of the supporting branch with probability $\frac{u}{n}$. Hence, the number of supporting stake drawn using VRFs in a single round is binomial with sample size $n$ and probability $u/n \cdot q/n$. The corresponding variance is 
$$
\text{Var}_{\textsc{BIN}} = n \frac{uq}{n^2} \left(1 - \frac{uq}{n^2}\right).
$$ 
The difference between these two quantities is illustrated in \autoref{tab:means_and_variances}, which contains for different selections of $n$, $q$, and $u$ the mean of $X$ (which equals $qu/n$ for both distributions), the variances of the hypergeometric distribution ($\text{Var}_{\textsc{HPG}}$) and the binomial distribution ($\text{Var}_{\textsc{BIN}}$), and the ratio of the latter to the former. 
If we keep $u/s$ constant at $\frac{2}{3}$, we see that the variance increases as $n$ or $q$ increases, but that $q$ has the biggest impact. 
We observe that the variance of the binomial distribution is consistently between 3 and 4 times larger than for the hypergeometric distribution.

\begin{table}[t]
\caption{Means and variances of the two different probability distributions (hypergeometric and binomial) for $X$ for different parameters.}
\label{tab:means_and_variances}
\begin{center}
\footnotesize
\begin{tabular}{ccc|c|ccc}
$n$ & $u$ & $q$ & mean & $\text{Var}_{\textsc{HPG}}$ & $\text{Var}_{\textsc{BIN}}$ & ratio\\ \toprule
150 & 100 & 15 & 10.0 & 3.02 & 9.33 & 3.09 \\
150 & 100 & 75 & 50.0 & 8.39 & 33.33 & 3.97 \\
1500 & 1000 & 75 & 50.0 & 15.84 & 48.33 & 3.05 \\
15000 & 10000 & 75 & 50.0 & 16.58 & 49.83 & 3.0 \\
1500 & 1000 & 750 & 500.0 & 83.39 & 333.33 & 4.0 \\
15000 & 10000 & 2000 & 1333.33 & 385.21 & 1214.81 & 3.15 \\
\end{tabular}
\end{center}
\end{table}%

\end{document}